\documentclass[runningheads]{llncs}
\usepackage[T1]{fontenc}
\usepackage{graphicx}
\usepackage{braket}
\usepackage{graphics}
\usepackage{amssymb}
\usepackage{amsmath, framed}
\usepackage{braket}
\usepackage[ruled,vlined,linesnumbered]{algorithm2e}
\SetKwComment{tcp}{\small$\triangleright$\,}{}
\usepackage{epsfig}
\usepackage{amsmath}
\usepackage{amsfonts}
\usepackage{amssymb}
\usepackage{enumerate}
\usepackage{booktabs}
\usepackage{multirow}
\usepackage{tabularx}
\usepackage{array}
\usepackage{enumitem}

\usepackage{circuitikz}

\usepackage{color,soul}

\usepackage{pdflscape}
\usepackage{longtable}
\usepackage{geometry}
\usepackage{graphicx}
\usepackage{fancyhdr} 

\newcommand{\Tr}{\text{Tr}}
\usepackage{amsmath, amssymb}
\newcommand{\Mcal}{\mathcal{M}}
\newcommand{\Hcal}{\mathcal{H}}

\usepackage{ifpdf}
\usepackage{qcircuit}
\usepackage{pifont}
\usepackage{graphicx}
\usepackage{braket}
\usepackage{pifont}
\usepackage{array}
\usepackage{longtable}
\usepackage{bm}
\usepackage{hyperref}
\usepackage{color}
\usepackage{makeidx}
\usepackage{cite}
\usepackage{mathtools}
\usepackage{bbm}
\usepackage{wrapfig}

\usepackage{amsmath, amssymb}    
\usepackage{tabularx}            
\usepackage{booktabs}            

\usepackage{tikz}                
\usetikzlibrary{arrows.meta}     
\usetikzlibrary{positioning}     

\usepackage{xcolor}              
\begin{document}
\title{Quantum AIXI: Universal Intelligence via Quantum Information}
\titlerunning{Quantum AIXI}
%
\author{Elija Perrier\inst{1}\orcidID{0000-0002-6052-6798} }
\authorrunning{E. Perrier}
\institute{Centre for Quantum Software \& Information, UTS, Sydney \email{elija.perrier@gmail.com}}

\maketitle              
\begin{abstract}
AIXI is a widely studied model of artificial general intelligence (AGI) based upon principles of induction and reinforcement learning. However, AIXI is fundamentally classical in nature - as are the environments in which it is modelled. Given the universe is quantum mechanical in nature and the exponential overhead required to simulate quantum mechanical systems classically, the question arises as to whether there are quantum mechanical analogues of AIXI. To address this question, we extend the framework to quantum information and present Quantum AIXI (QAIXI). We introduce a model of quantum agent/environment interaction based upon quantum and classical registers and channels, showing how quantum AIXI agents may take both classical and quantum actions. We formulate the key components of AIXI in quantum information terms, extending previous research on quantum Kolmogorov complexity and a QAIXI value function. We discuss conditions and limitations upon quantum Solomonoff induction and show how contextuality fundamentally affects QAIXI models.

\end{abstract}
\section{Introduction}
A cornerstone in the theoretical landscape of AGI is the AIXI model \cite{hutter2004universal,hutter2007}, which provides a mathematically rigorous and complete framework for an optimal Bayesian reinforcement learning agent. AIXI's optimality is rooted in Solomonoff's theory of universal induction \cite{solomonoff1964} and classical computability theory. However, AIXI and its underlying assumptions are fundamentally classical. Given that our universe is governed by quantum mechanics, a fundamental question arises: what constitutes universal intelligence in a quantum mechanical world? This paper addresses this question by developing the mathematical foundations for Quantum AIXI (QAIXI), an AGI model operating within the principles of quantum mechanics. Using the formal theory of quantum information processing, we extend prior work on quantum intelligence \cite{ozkural2015ultimate,catt2020} and situates QAIXI as an agent anchored in both quantum and classical registers interacting with the environment via quantum and classical channels. Using this formulation, we contribute the following:
\begin{enumerate}[label=(\alph*)]
\item Describing the QAIXI agent's interaction loop using density operators and quantum measurement theory along with specifying quantum Kolmogorov complexity ($K_Q$) \cite{berthiaume2001quantum,vitanyi2002quantum} in terms of quantum information channels and its role in a universal quantum prior.
    \item Utilising the theory of quantum-to-quantum (QTQ) and quantum-to-classical channels \cite{watrous_theory_2018} to describe the space of possible quantum environments.
    \item Formulating a quantum equivalent of the universal Bayesian mixture over quantum environments while adapting a quantum Bellman equation and the QAIXI policy based on quantum formalism.
    \item Describing circumstances in which the Kochen-Specker theorem \cite{kochenSpecker1967} renders any QAIXI history fundamentally contextual.
\end{enumerate}
\vspace{-0.5em}
QAIXI is impractical to implement. Our contribution therefore aims to contribute to debate over the theoretical models of optimal intelligence that account for the fundamentally quantum nature of the universe, particularly limitations stemming from the inherent properties of quantum state spaces, quantum computation, and quantum measurement. Appendices can be found via \cite{perrier2025qaxirefs}.
\vspace{-2em}

\subsection{Related Work}
\vspace{-1em}
AIXI \cite{hutter2004universal,hutter2007,hutter2010} is among the leading AGI proposals at the centre of debates over AGI, as both a model of universal intelligence and a proposal against which other proposals are assessed.  It takes Solomonoff’s universal induction, folds it into sequential decision theory, and produces what is arguably the cleanest statement of an unbounded optimal agent.  Subsequent symbolic, connectionist, and hybrid AGI proposals typically position themselves by (i) trying to approximate AIXI in practice (e.g.\ AIXItl, MC-AIXI, information-geometric variants) or (ii) critiquing AIXI’s reliance on classical assumptions and incomputable quantities.  Logicist programmes \cite{bennettmaruyama2022b}, emergentist lines \cite{sole2019}, distributed cognitive architectures \cite{goertzel2021,goertzel2023}, and self-organising approaches \cite{mcmillen2024} all share AIXI’s basic agent–environment model, albeit while extending, varying or removing its component Bayesian machinery.  
AIXI has been subject to critical review over the last two decades on a number of grounds: (1) physical unrealisability or super-Turing requirements \cite{wang2015assumptions,leike2015}; (2) the Cartesian boundary between software and hardware \cite{bennett2023b,bennett2024appendices}; (3) inadequate treatment of resource constraints, counterfactual reasoning or multi-agent reflection \cite{soares2015two,fallenstein2015reflective}; and (4) challenges to the Kolmogorov-based prior itself \cite{thorisson2016artificial}.  Even so, AIXI remains a cornerstone of classical AGI (CAGI) theories. The advent of quantum information technologies (QIT) \cite{NielsenChuang2010,aaronson2013quantum} and interest in their computational capabilities and quantum forms of AGI (QAGI) poses a natural question for AIXI models: if AIXI purports to be a universal agent in a \emph{classical} world, what replaces it in a \emph{quantum} universe?  Most quantum-AI research to date is narrowly cast, focusing on specific technical features such as quantum decision theory, quantum machine learning \cite{schuld_machine_2021,BiamonteEtAl2017,DunjkoBriegel2018,cerezo_challenges_2022,perrier2024quantum}, quantum reinforcement learning \cite{DongEtAl2006,jerbi_parametrized_2021,meyer2022survey}, and hybrid variational schemes \cite{schuld_evaluating_2019,SchuldPetruccione2018} (including tomography \cite{sarkar2022qksa}).  These tend to involve classical agents accessing quantum algorithms providing prospective quantum advantage (e.g.\ amplitude-estimation–based policy evaluation \cite{liu_rigorous_2021}), albeit with ongoing debate about their practical reach \cite{aaronson2014quantum}.  

Attempts to synthesise AIXI with quantum mechanics (and quantum information processing) have been relatively limited. Work on Solomonoff induction \cite{solomonoff1964} and quantum computation (e.g. Deutsch’s quantum Turing machines (QTM) \cite{deutsch1985}) have considered component features of the AIXI program, including substituting out specific classical components for quantum analogues (e.g. replacing classical with quantum Grover search \cite{catt2020}). Yet quantum mechanics operationally - and ontologically - carries profound differences from classical ontology (and computation). Quantum systems may subsist in superposition states, be entangled and lack definitive identity until measurement interaction. Moreover, any integration of quantum mechanics and AIXI must also reckon with seminal results in quantum foundations: contextuality, non-locality, and quantum measurement — all absent from classical AIXI.  Work on quantum causality \cite{brukner2014quantum,wiseman2016causarum}, algorithmic thermodynamics \cite{ebtekar2025foundations}, and many-worlds decision theory \cite{Wallace2012} offer contributions on foundational quantum mechanical issues facing quantum AIXI-style agents whose quantum operations disturb the very environment they probe.  At present, however, there is neither a consensus quantum prior (the role played by $2^{-\!K(\nu)}$ in AIXI) nor a tractable complexity class in which QAIXI could be approximated.

\section{Classical AIXI}
The AIXI agent \cite{hutter2004universal} is a theoretical model for a universally optimal reinforcement learning agent. It interacts with an unknown environment $\mu$ in cycles. In each cycle $t$, the agent chooses an action $a_t \in \mathcal{A}$ and receives a percept $e_t=(o_t, r_t) \in \mathcal{O} \times [0,1]$, consisting of an observation $o_t$ and a reward $r_t$. The history is $h_{<t} = a_1e_1\dots a_{t-1}e_{t-1}$. AIXI's policy $\pi^{AIXI}$ aims to maximise future expected rewards:
\begin{equation}\label{eq:AIXI_classical_redef}
  a_t := \pi^{\mathrm{AIXI}}(h_{<t})
  = \arg\max_{a_t}
    \sum_{\nu\in\Mcal_U} w_\nu(h_{<t})
    \sum_{e_t}\max_{a_{t+1}}\sum_{e_{t+1}}\dots\max_{a_m}\sum_{e_m}
        \left[\;\sum_{k=t}^{m}\gamma^{k-t}r_k\right]
        \prod_{j=t}^{m}\nu(e_j\mid h_{<j}a_j).
\end{equation}
where $m$ is the lifetime of the agent, $\gamma \in [0,1)$ is a discount factor, and the outer sum is over all environments $\nu$ in a universal class $\Mcal_U$ (e.g., all chronological semi-computable environments $\Mcal_{sol}$), weighted by the posterior $w_\nu(h_{<t}) \propto 2^{-K(\nu)}$ (see Table \ref{tab:symbols} in the Appendix). $K(\nu)$ is the Kolmogorov complexity of the classical Turing Machine (CTM) describing environment $\nu$. The sum over environments constitutes Solomonoff's universal prior $\xi_U$:
\begin{equation} \label{eq:xi_U_classical_redef}
    \xi_U(e_{1:m} || a_{1:m}) := \sum_{\nu \in \Mcal_{sol}} 2^{-K(\nu)} \nu(e_{1:m} || a_{1:m}).
\end{equation}
AIXI is Pareto-optimal and self-optimising in the sense of \cite{hutter2004universal}. However, it is incomputable due to the use of $K(\cdot)$ and the sum over all CTMs. Its ontological assumptions are classical: deterministic or classically stochastic environments, objective histories, and classical information theory. 
\vspace{-1em}

\section{Quantum Computational Foundations}
\vspace{-0.5em}
\subsection{Quantum information}
To construct QAIXI, we must replace classical computational notions with their quantum counterparts. However, this is not a straightforward isomorphic mapping. The unique properties of quantum mechanics problematise conventional assumptions underlying classical AIXI, such as agent identity, definitiveness of state and the separability and distinguishability of an agent from its environment. To formulate QAIXI in a way that accommodates and caters for these ontological differences, we turn to the language and formalism of quantum information theory \cite{watrous_theory_2018}. In QIP, \textit{systems} (agents) and \textit{environments} are described in terms of \textit{registers} $\texttt{X}$ (e.g. bits) comprising information drawn from a classical alphabet $\Sigma$. Registers may be in either classical or quantum states. We define a Hilbert space $\mathcal X=\mathbb C^{|\Sigma|}$ with computational basis 
$\{|s\rangle\}_{s\in\Sigma}$. A \textit{quantum state} of a register \texttt{X} (associated with space $\mathcal{X}$) is a density operator $\rho \in \mathcal{D}(\mathcal{X})$, i.e., a positive semi-definite operator with $\Tr(\rho)=1$. Classical states are then $\rho$ which are diagonal in the basis of $\mathcal{X}$. Interactions with (and changes to) states occur via \textit{channels} which are superoperators. They define how quantum and classical states interact. Classical-to-classical registers (CTC) preserve classical states, classical-to-quantum (CTQ) channels encode classical information in quantum states, quantum-to-classical (QTC) channels extract classical information from quantum states, decohering them in the process; quantum-to-quantum (QTQ) channels form coherent (e.g. unitary) transformations between quantum registers. In this framing, both agents and environments are registers which may be CAGI (classical state sets) or QAGI (quantum state sets). They may interact in ways that are coherent (quantum-preserving) or classical (via CTC or QTC maps).

\subsection{Quantum AIXI} \label{sec:qagi-models}
The components of QAIXI can then be understood as follows. Let the QAIXI agent be associated with a quantum register \texttt{A} representing the internal degrees of freedom, while its environment is represented by a quantum register \texttt{E}.  Their corresponding (finite–dimensional) complex Hilbert spaces are denoted by $\mathcal H_{A}$ and $\mathcal H_{E}$.  At interaction step~$t$ the agent's private state is a density operator $\rho_{A}^{(t)}\in\mathcal D(\mathcal H_{A})$ that may encode its complete history, its current belief state, or—in the ideal case—an explicit representation of the universal quantum mixture~$\Xi_{Q}$.  The environment is simultaneously described by $\rho_{E}^{(t)}\in\mathcal D(\mathcal H_{E})$.  Taken together we assume the composite system occupies the joint state $\rho_{AE}^{(t)}\in\mathcal D(\mathcal H_{A}\otimes\mathcal H_{E})$. Note that $\rho_{AE}^{(t)}$ can be entangled, so that $\rho_{AE}^{(t)}\neq\rho_{A}^{(t)}\otimes\rho_{E}^{(t)}$ (see App. \ref{app:entangled-loop}). A history is therefore an operator‐valued stochastic process rather than a sequence of point events.
\\
\\
\textbf{Actions}.
At each cycle the agent chooses an \emph{action} $a_{t}$, formally a completely–positive, trace–preserving (CPTP) map acting on the environment register (or on a designated subsystem of the joint register).  The two canonical cases are as follows.  
\begin{enumerate}
    \item First, when the agent performs a coherent control operation, $a_{t}$ is realised by a unitary channel $\Phi_{U_{a_{t}}}\!:X\mapsto U_{a_{t}}\,X\,U_{a_{t}}^{\dagger}$, applied either on $\mathcal H_{E}$ alone or on $\mathcal H_{A}\otimes\mathcal H_{E}$.  Such a map preserves superposition and entanglement and therefore belongs to the quantum‑to‑quantum (QTQ) class of channels.
    \item Second, the agent may decide to interrogate the environment by means of a \emph{quantum instrument} $\mathcal I_{a_{t}}=\{\mathcal E^{a_{t}}_{k}\}_{k\in\Gamma_{\mathrm{obs}}}$, where the outcome alphabet $\Gamma_{\mathrm{obs}}$ is classical.  Each branch map $\mathcal E^{a_{t}}_{k}:\mathcal L(\mathcal H_{E'})\to\mathcal L(\mathcal H_{E'})$ is completely positive and trace non–increasing, satisfies $\sum_{k}\mathcal E^{a_{t}}_{k} = \mathcal E^{a_{t}}$ (their sum is a CPTP map) and $\mathcal E^{a_{t}\dagger}_{k}(\mathbb I)=\mathbb I$, and typically takes the form $\mathcal E^{a_{t}}_{k}(X)=M^{a_{t}}_{k}\,X\,M^{a_{t}}_{k}\!{}^{\dagger}$ for a POVM $\{M^{a_{t}}_{k}\}$.  Because every branch deposits a classical record $k$ while decohering the measured subsystem, $\mathcal I_{a_{t}}$ is a paradigmatic quantum‑to‑classical (QTC) channel.
\end{enumerate}
Formally the action $a_t$ can be represented via $\Phi_{a_{t}}\;:\;\mathcal L(\mathcal H_{AE})\;\longrightarrow\;\mathcal L(\mathcal H_{AE})$ which may be indexed by the form of $\mathcal L$.
\\
\\
\textbf{Percepts and rewards}.  What the QAIXI agent perceives is determined entirely by the instrument it applies.  A measurement outcome $k\in\Gamma_{\mathrm{obs}}$ becomes the observation component $o_{t}$ of the percept $e_{t}=(o_{t},r_{t})$.  The outcome is drawn with probability: 
\[
  \Pr(k) := \Pr\bigl(o_{t}=k\mid a_{t},\rho_{AE}^{(t-1)}\bigr)
    =\operatorname{Tr}\bigl[\mathcal E^{a_{t}}_{k}\bigl(\operatorname{Tr}_{A}\rho_{AE}^{(t-1)}\bigr)\bigr].
\]
The reward $r_{t}\in\Gamma_{\mathrm{rew}}$ (with $\Gamma_{\mathrm{rew}}\subseteq\mathbb R$ in the standard reinforcement‑learning setting) is computed by a classical post‑processing function that may depend on both $o_{t}$ and the agent's prior internal state $\rho_{A}^{(t-1)}$.  The resulting data $o_{t}$ and $r_{t}$ are stored in designated registers - but it is at this stage \textit{classical} data. So it can be stored in a classical register via a CTC map, or encoded in a quantum register via a CTQ map. 
\\
\\
\textbf{Interaction loop}. The QAIXI cycle of interaction with the environment is as follows.  Given the pre‑interaction state $\rho_{AE}^{(t-1)}$, the agent selects $a_{t}$.  If $a_{t}$ is a unitary, the composite state updates coherently to $\rho_{AE}^{(t)}=\Phi_{U_{a_{t}}}\bigl(\rho_{AE}^{(t-1)}\bigr)$.  If instead $a_{t}$ is an instrument, an outcome $k$ is observed with the above probability, the environment becomes: 
\begin{align}
  \rho_{E'}^{(t)}=\Tr_A \rho_{AE}^{(t)}
    =\frac{\mathcal E^{a_{t}}_{k}\bigl(\operatorname{Tr}_{A}\rho_{AE}^{(t-1)}\bigr)}{\Pr(k)}.
\end{align}
The global post‑measurement state is $\rho_{AE}^{(t)}=\rho_{A}^{(t-1)}\otimes\rho_{E'}^{(t)}$ (assuming no entanglement - see Appendix \ref{app:entangled-loop} for the general case).  Finally, the agent applies an internal CPTP map—its update rule—to obtain $\rho_{A}^{(t)}$ from $\rho_{A}^{(t-1)}$ in the light of $(a_{t},o_{t},r_{t})$.  This update may be trivial (identity) if $\rho_{A}^{(t)}$ is purely classical or it may itself be a non‑trivial quantum channel when the agent maintains coherent beliefs. The agent’s internal memory is refreshed by a chosen CPTP map $\mathsf U_{\!\mathrm{int}}:\rho_{A}^{(t-1)}\mapsto\rho_{A}^{(t)}$ that may itself depend on $(a_{t},k,r_{t})$.  This operational definition fixes the concepts of \emph{actions}, \emph{percepts}, and \emph{rewards} in the quantum setting, providing the concrete substrate on which the universal mixture $\xi_{Q}$ and the QAIXI value functional are built.
\vspace{-1em}

\subsection{Quantum Kolmogorov complexity} QAIXI relies upon quantum \cite{berthiaume2001quantum} (rather than classical) Kolmogorov complexity. Let $\mathcal Q_{\!\mathrm{sol}}$ be the set of all chronological, semi–computable quantum environments, each such environment $Q\!\in\!\mathcal Q_{\!\mathrm{sol}}$ being represented by a QTM that outputs an instrument sequence. An \emph{environment} $Q$ is a CPTP map acting on a register
$\mathcal{H}_E$. We represent each quantum environment
\(Q:\mathcal L(\mathcal H_E)\!\to\!\mathcal L(\mathcal H_E)\)
by its Choi–Jamiolkowski vector  \cite{haapasalo2021choi} to identify the channel $Q: \mathcal{L}(\mathcal H_E) \to (\mathcal H_E)$ with the purified vector
\(
\bigl|Q\bigr\rangle
:= (\mathbb{I}\otimes Q)\bigl|\Phi^+\bigr\rangle
\in \mathcal{H}_E^{\otimes 2},
\)
where $|\Phi^+\rangle =
\frac1{\sqrt{d}}\sum_{i=1}^{d}|i\,i\rangle$
for $d=\dim\mathcal{H}_E$.  All trace-distance
bounds on $|Q\rangle$ translate to diamond-norm bounds on $Q$. Its quantum Kolmogorov complexity is: 
\begin{equation}
\label{eq:QKC}
K_Q(Q)\;:=\;
\min_{p\in\{0,1\}^\star}\Bigl\{|p| \;\Bigl|\;
\bigl\|U_{\text{univ}}\,|p\rangle|0\rangle -
|Q\rangle\otimes|{\rm aux}\rangle \bigr\| \le \varepsilon\Bigr\},
\end{equation}
for a universal QTM, $U_{\text{univ}}$, and fixed $\varepsilon<1$ with Hilbert-Schmidt norm $||\cdot||$. Equation~\eqref{eq:QKC} therefore
measures program length needed to approximate the channel $Q$. The $|\mathrm{aux}\rangle$ term is an ancilla whose dimension is at most
polynomial in $|p|$. Here we have assumed that for any two universal QTMs $U_1,U_2$ there exists a constant
$c_{U_1,U_2}$ such that for every environment $Q$ we have $\bigl|K^{U_1}_Q(Q)-K^{U_2}_Q(Q)\bigr|\;\le\;c_{U_1,U_2}$.
The quantum Solomonoff mixture is the semi–density operator:
\begin{equation}\label{eq:xiQ‐density}
  \Xi_{Q}\bigl(a_{1:m}\bigr)
  \;:=\;\sum_{Q\in\mathcal Q_{\!\mathrm{sol}}}2^{-K_{Q}(Q)}\,\rho^{Q}_{E}(a_{1:m}),
\end{equation}
where $\rho^{Q}_{E}(a_{1:m})$ is the environment state generated by $Q$ under the action sequence $a_{1:m}$ and $\Tr \Xi_{Q}\bigl(a_{1:m}\bigr) \leq 1$.
Projecting $\Xi_{Q}$ onto any classical POVM recovers the probability mixture
$\xi_{Q}(e_{1:m}\Vert a_{1:m})$ that generalises \eqref{eq:xi_U_classical_redef}. Whenever every admissible environment $Q$ outputs commuting observables and the QAIXI agent restricts itself to instruments diagonal in that basis, each $\rho^{Q}_{E}(a_{1:m})$ becomes a diagonal density operator encoding an ordinary probability measure.  In that limit $K_{Q}(Q)$ coincides (up to a constant) with the classical Kolmogorov complexity $K(\nu)$ of the induced CTM $\nu$, while $\xi_Q$ reduces to the classical Solomonoff prior $\xi_{U}$ and the AIXI policy.  Hence the quantum formulation generalises the classical one (for a discussion of complexity and resource-constraints, see \cite{ozkural2016ultimate}).

\subsection{QAIXI value functional}  We also define the quantum analogue of the Bellman equation. For a policy $\pi$ and environment $Q$ define the discounted return
\begin{equation}
  V^{\pi}_{Q}\bigl(\rho_{AE}^{(t-1)}\bigr)
  =\;\mathbb E_{\pi,Q}\Bigl[\,\sum_{k=t}^{m}\gamma^{k-t}r_{k}\;\Bigm|\;\rho_{AE}^{(t-1)}\Bigr].
\end{equation}
Because each branch map $E^{a_t}_k$ both yields $o_t$ and
updates $\rho^{(t)}_E(k)$, the expectation
$E_{\pi,Q}\!\left[\sum_{k=t}^m\gamma^{k-t}r_k\right]$ is taken over the (non-Markovian) instrument-conditioned trajectory measure, not over an
i.i.d.\ sequence. It terminates at finite horizon $m$.
Averaging over the posterior universal mixture gives:
\begin{equation}\label{eq:qaixi‐value}
  V^{\pi}_{\Xi_Q}\bigl(\rho_{AE}^{(t-1)}\bigr)
  =\sum_{Q\in\mathcal Q_{\!\mathrm{sol}}}
       w_Q(h_{<t})\;
       V^{\pi}_{Q}\bigl(\rho_{AE}^{(t-1)}\bigr),
\end{equation}
where $w_Q(h_{<t}) \propto 2^{-K_Q(Q)}     \prod_{i=1}^{t-1}\! \nu_Q(e_i\mid h_{<i}a_i)$ is the normalised Bayesian weight after $t-1$ cycles (with $\nu_Q$ the classical probability distribution from measuring $Q$). The \emph{QAIXI policy} maximises this functional at every step:
\begin{equation}\label{eq:qaixi-policy}
  a_t := \pi^{\text{QAIXI}}\bigl(\rho_{AE}^{(t-1)}\bigr)
  = \arg\max_{a_t\in\mathsf{Act}}
      \mathbb E_{k\sim p_t(\cdot\mid a_t)}
        \Bigl[r_t(k)+\gamma\,V^{\text{QAIXI}}_{\Xi_Q}\bigl(\rho_{AE}^{(t)}(k)\bigr)\Bigr],
\end{equation}
Here $\rho^{(t)}_{AE}(k)$ is the post‑measurement state with probabilities as $p_{t}(k \mid a_{t})
  \;=\;
  \operatorname{Tr}\!\bigl[  \mathcal{E}^{a_{t}}_{k}\bigl(\,\Xi_{Q}^{(t)}(a_{1:t-1})\bigr)
  \bigr]$.  The fix‑point equations implicit in Eqns. (\ref{eq:qaixi‐value}-\ref{eq:qaixi-policy}) form the \emph{quantum Bellman equation}. Moreover, the act of observation fundamentally alters $\rho_E$ due to quantum back-action (hence the need to use the updated state). Note that as we are dealing with quantum trajectories, the classical notion of histories $h_{<t}$ must now account for acts of measurement (via instruments) themselves: different measurement choices lead to different post-measurement states and therefore different future dynamics.
The substitution $K\!\to\!K_{Q}$ still leaves the mixture \eqref{eq:xiQ‐density} incomputable. It also introduces an additional obstacle: evaluating $\rho^{Q}_{E}(a_{1:m})$ may be computationally hard even for $m=1$. In certain cases, preparing an arbitrary $n$-qubit quantum state $\ket{\psi}$ may require a quantum circuit of size exponential in $n$ \cite{knill1995approximation}. This implies that $K_Q(\ket{\psi})$ can be exponentially larger (in qubit count) than $n$ for complex states, unlike classical $K(x)$ which is at most the uncompressed bit-length $\ell(x) + c$ (assuming $x$ is the program itself, not an arbitrary binary string) \cite{li2019introduction}. The upper bound is still $O(2^n)$ bits. Practical approximations therefore hinge on identifying structured subclasses of $\mathcal Q_{\!\mathrm{sol}}$ (e.g. stabiliser processes or tensor‑network environments) for which both $K_{Q}$ and the Born probabilities are efficiently approximable.
\vspace{-1em}
%
%
\section{Quantum Solomonoff Induction (QSI)}
\label{sec:qsi}
\vspace{-0.5em}
Classical Solomonoff induction assigns to every finite data string a
universal a-priori probability obtained by summing, with
complexity–based weights, over all computable environments that could
have produced that string.  In our quantum setting the data arriving at
an agent are the classical outcomes of \emph{instrument branches}
executed on a quantum environment.  Quantum Solomonoff
induction therefore replaces probability measures by
\emph{semi-density operators} and Kolmogorov codes by \emph{quantum
program states}.  
\\
\\
\textbf{Environment class and universal mixture}. 
Let $\mathcal Q_{\!\mathrm{sol}}$ denote the set of all chronological,
semi-computable quantum environments (quantum channels generated sequentially by a QTM).  Each
$Q\in\mathcal Q_{\!\mathrm{sol}}$ is specified by a program
$p(Q)\in\{0,1\}^{*}$ for a fixed universal QTM
$U_{\mathrm{univ}}$. Running $p(Q)$ produces, step by step, a sequence
of CPTP maps (quantum channels) that act on the environment register and a POVM \emph{measurement specification} for the classical transcript. Here $Q$ plays the role of $\nu$ in the classical case.  Its
quantum Kolmogorov complexity is
\(
  K_{Q}(Q):=\min\{|p(Q)|:U_{\mathrm{univ}}(p(Q))\simeq Q\}.
\)
For an action sequence $a_{1:m}$ we write
\(
  \rho_{E}^{Q}(a_{1:m})\in\mathcal D(\mathcal H_{E})
\)
for the (generally mixed) state that $Q$ prepares on the environment
register immediately before the observation at cycle~$m$. The \emph{universal semi-density operator} conditional on the agent’s
actions is then:
\begin{equation}
  \label{eq:XiQ-density}
  \Xi_{Q}(a_{1:m})
  :=
  \sum_{Q\in\mathcal Q_{\!\mathrm{sol}}}
    2^{-K_{Q}(Q)}\,
    \rho_{E}^{Q}(a_{1:m}),
  \qquad
  0<\Tr\bigl[\Xi_{Q}(a_{1:m})\bigr]\le 1.
\end{equation}
Projecting $\Xi_{Q}$ onto the POVM that
implements the agent’s instrument at cycle~$m$ yields the scalar
\(
  \xi_{Q}(e_{1:m}\Vert a_{1:m})
  :=
  \Tr\!\bigl[
        M_{e_{1:m}}\,\Xi_{Q}(a_{1:m})
      \bigr],
\)
which reduces to the classical Solomonoff prior
$\xi_{U}$ when all $M_{e_{1:m}}$ commute. The trace gives the probability of seeing that outcome sequence given the action sequence $a_{1:m}$.
\\
\\
\textbf{Bayesian updates}. Given a history
$h_{<t}=(a_{1:t-1},e_{1:t-1})$, the posterior semi-density operator is
obtained by the update followed by a renormalisation:
\begin{equation}
  \label{eq:posteriorupdate}
  \Xi_{Q}^{(t)}\!(a_{1:t-1})
  :=
  \frac{
     \mathcal M_{e_{t-1}}\!\bigl(\Xi_{Q}^{(t-1)}(a_{1:t-2})\bigr)
  }{\Tr[
        \mathcal M_{e_{t-1}}\!\bigl(\Xi_{Q}^{(t-1)}(a_{1:t-2})\bigr)
      ]},
  \qquad
  \Xi_{Q}^{(0)}:=\Xi_{Q}.
\end{equation}
Here $\mathcal M_{e_{t-1}}$ is the CP map of the realised measurement branch of the
agent’s instrument at step~$t-1$.  The distribution for the
next observation is:
\begin{align}
  \xi_{Q}(\,\cdot\,\Vert h_{<t},a_t)
  =k\mapsto
   \Tr\bigl[\mathcal E^{a_t}_{k}(
             \Xi_{Q}^{(t)}(a_{1:t-1})
           )\bigr].
\end{align}

\subsection{Convergence theorem}
Analogous to classical Solomonoff induction, under certain conditions we might expect QSI to exhibit convergence properties. The convergence properties of QSI are complex and remain an open question. To sketch out the issues, we consider the following model (see Appendix \ref{app:qsi-convergence-elaboration} for more detail). Write $D(\rho\Vert\sigma)=\Tr[\rho(\ln\rho-\ln\sigma)]$ for the
\emph{Umegaki relative entropy}.  Fix a true quantum environment
$Q^{\star}\in\mathcal Q_{\!\mathrm{sol}}$ and let
$\rho_{E}^{\star}(a_{1:m})$ be its state sequence.  We assume:

\smallskip
\noindent
\emph{(C1) Ergodicity.}  
There is a $\delta>0$ such that for every admissible action policy the
time-averaged state satisfies:
$
  \liminf_{m\to\infty}
    \tfrac1m\sum_{k=1}^{m}
      \bigl\lVert\rho_{E}^{\star}(a_{1:k})-
                     \rho_{E}^{\star}(a_{1:k-1})\bigr\rVert_{1}
  \leq\delta.
$

\noindent
\emph{(C2) Informational completeness.}  
Each cycle’s instrument has a POVM refinement whose classical
Fisher information matrix is full-rank up to error $\epsilon>0$.

\noindent
\emph{(C3) Complexity gap finite.}  
$K_{Q}(Q^{\star})<\infty$ and
$
  g:=\sum\nolimits_{Q\neq Q^{\star}}
       2^{-\bigl(K_{Q}(Q)-K_{Q}(Q^{\star})\bigr)}
  <\infty.
$

\begin{theorem}[QSI convergence]
\label{thm:qsi-convergence}
Under \emph{(C1)}–\emph{(C3)} the posterior density operator satisfies
\begin{align}
  \mathbb E_{Q^{\star}}\!\bigl[
     D\!\bigl(
        \rho_{E}^{\star}(a_{1:t})
        \,\big\Vert\,
        \Xi_{Q}^{(t)}(a_{1:t})
     \bigr)
  \bigr]
  \;\le\;
  \frac{K_{Q}(Q^{\star})\ln 2+\ln(1+g)}{\!\!t}. \label{eqn:QSI:posteriordensity}
\end{align}
Consequently, by the quantum Pinsker inequality \cite{csiszar2011information,hirota2020application} (total variation distance given by KL-divergence),
\begin{align}
  \mathbb E_{Q^{\star}}\!\bigl[
     \tfrac12
     \bigl\lVert
       \rho_{E}^{\star}(a_{1:t})-\Xi_{Q}^{(t)}(a_{1:t})
     \bigr\rVert_{1}
  \bigr]
  =\mathcal O\!\bigl(t^{-1/2}\bigr).
\end{align}
\end{theorem}
\textit{Sketch of proof}.
Proving QSI convergence is an open question. One potential avenue is as follows. Define likelihood operators
$
  \Lambda_{t}^{Q}
  :=\mathcal M_{e_{t}}\circ\dots\circ\mathcal M_{e_{1}}
    \bigl(\rho_{E}^{Q}(a_{1:t})\bigr).
$, with $\Lambda^\Xi_k$ defined similarly substituting in $\Xi^{(0)}_{Q}(a_{1:t})$. $\Tr\Lambda_{t}^{Q}$ gives the joint Born probabilities $\Pr_Q(e_1,...,e_t\Vert a_1,...,a_t)$ reflecting the likelihood of the observed trajectory. Akin to the classical Solomonoff case, monotonicity of quantum relative entropy (see \cite{watrous_theory_2018}) implies that applying any
branch map $\mathcal M_{e_k}$ can only \emph{decrease} divergence i.e. $D\!\bigl(\rho_{E}^{\star}\,\Vert\,\Xi_Q\bigr)
      \;-\;
D\!\bigl(\Lambda^{\star}_{k}\,\Vert\,\Lambda^{\Xi}_{k}\bigr)$ is $\mathbb E_{Q^{\star}}\!\text{-martingale}$.
Condition (C2) guarantees that each step’s expected drop is
non-negative.  Summing these expected drops from $k=1$ up to $k=t$
therefore shows that the \emph{average} divergence after $t$ cycles is
bounded by the initial one,
\[
\mathbb E_{Q^{\star}}\!\Bigl[
     D\!\bigl(\rho_{E}^{\star}\,\Vert\,\Xi_Q\bigr)
\Bigr]
  \;\le\;
\frac{K_Q(Q^{\star})\ln 2+\ln(1+g)}{t}.
\]
Finally, the quantum Pinsker inequality \cite{csiszar2011information} turns this $1/t$ bound on
relative entropy into an $\mathcal O(t^{-1/2})$ bound on trace
distance, completing the convergence claim. The replacement of full informational completeness by
$\epsilon$-completeness introduces an extra $O(\epsilon)$ term in the
Pinsker bound, which vanishes as $\epsilon\!\to\!0$.
Equation~\eqref{eq:posteriorupdate} is executed by the internal
belief-revision QTQ channel $\mathsf U_{\!\mathrm{int}}$.  The posterior $\Xi_{Q}^{(t)}$ supplies
the conditional expectation required for the quantum Bellman
equation.  QSI is the inductive
basis of QAIXI concentrating weight on environments that
remain compatible with classical measurements. In the commuting limit the theorem reduces to the classical Solomonoff
convergence bound~\cite{hutter2004universal}. A formal proof would need to justify martingale convergence for operator-valued universal mixtures under quantum measurement dynamics (App. \ref{app:qsi-convergence-elaboration}). \vspace{-1em} 

\subsection{QSI Limitations}
Aside from the computational complexity considerations of QSI (which we postpone for future work), several other significant challenges arise with QSI:
\begin{enumerate}
    \item \textit{Specifying $\mathcal Q_{\!\mathrm{sol}}$ and $K_{Q}(\cdot)$}. 
          A universal class of \emph{computable quantum environments} must be fixed, e.g.\ the set of all chronological, semi–computable QTMs that output instrument sequences.  
          The quantum description-length
          \(
            K_{Q}(Q)
          \)
          is then the length of the shortest qubit-string that drives a fixed universal QTM to approximate the environment $Q$ within trace-distance~$\varepsilon$.  
          Proving universality and invariance up to an additive constant—standard for classical Kolmogorov complexity—is technically subtler in the quantum setting because programs may themselves be in superposition.  Moreover, if an environment is weakly entangled, this could let a classical AIXI approximate QAIXI’s value in practice. 
    \item \textit{Measurement back-action}.  
          Born probabilities
          \(
            \Pr_{Q}(e_{1:m}\Vert a_{1:m})
          \)
          are computed from the iterated CPTP maps that model the interaction loop
          (Sec.~\ref{sec:qagi-models}).  
          Each measurement branch both yields the classical outcome $o_t$ and
          updates the post-measurement state  
          \(
            \rho_{E}^{(t)} \mapsto \rho_{E}^{(t+1)}(o_t)
          \). Because measurement alters $\rho_E$, the quantum back-action makes the expectation in the Bellman equation dependent on the sequence of quantum operations and classical outcomes in a complex, non-Markovian way, so the likelihood of future data depends on the entire history  
          \(h_{<t}a_t\).

    \item \textit{Non-locality and contextuality}.  
          When $Q$ prepares entangled states, the joint distribution  
          \(
            \Pr_{Q}(e_{1:m}\Vert a_{1:m})
          \)
          can violate Bell inequalities \cite{bell_speakable_2004} and exhibit contextual dependence on the full
          instrument sequence.  QSI must therefore aggregate over models that are not representable by any classical hidden-variable process, which complicates identifiability and convergence analyses.

    \item \textit{Incomputability and resource sensitivity of the prior}.  
          Like its classical counterpart, $K_{Q}$ is uncomputable.  
          Moreover, the weight  
          \(2^{-K_{Q}(Q)}\)  
          does not penalise the physical resources needed either to prepare the initial
          state of~$Q$ or to simulate its dynamics—tasks that can be exponential-time or generally infeasible.  
          For physically realistic prediction one may replace $K_{Q}$ by a resource-aware
          complexity measure that incorporates, for example, circuit depth or Trotter
          step count. 
          \item \textit{Quantum error correction (QEC)}. Implementing QAIXI on real, noisy hardware requires quantum error correction. Encoding each logical qubit in hundreds-to-thousands of physical qubits both (i) inflates the effective description-length, shifting the universal prior toward far simpler environments, and (ii) introduces continual syndrome-measurement back-action whose residual logical noise is only bounded below a code-specific threshold—together, these effects raise the resource bar and weaken the assumptions (ergodicity and informational completeness) behind potential convergence guarantees.
\end{enumerate}
Moreover, there are specific challenges to QAIXI and QSI raised by quantum foundational issues.\\
\\
\textbf{Bell non-locality}.  
      In any QAIXI environment that distributes entangled subsystems
      to space-like separated agent components the joint percept
      distribution can violate the CHSH inequality.  
      No classical hidden-variable environment
      $\nu$ can reproduce those statistics. A classical $\xi_{U}$ assigns zero mass to each Bell-violating local hypothesis, so only the quantum mixture $\xi_{Q}$ has non-trivial mass to contribute.  
      When the agent exploits Bell-type correlations for
      decision-making, its percepts are no longer
      conditionally-independent given the entire history, breaking the
      martingale structure assumed in Theorem \ref{thm:qsi-convergence}. See Appendix \ref{app:bell-no-cloning}.
\\
\\
\textbf{Kochen–Specker contextuality}.  
      If the true environment prepares a KS set of
      projectors, then there exists \emph{no} history-independent map
      $v:\mathcal P(\Hcal_{E})\!\to\!\{0,1\}$ assigning pre-existing
      outcomes to every measurement the agent may perform.  
      QSI must therefore sum over hypotheses whose outcome statistics
      depend on the entire future instrument sequence up to horizon $m$. Bayesian
      updates must track non-commuting observables (see Appendix \ref{app:contextuality}). We adapt this theorem for QAIXI below.

\begin{corollary}[QAIXI Contextuality]
\label{thm:contextual_predictor_full}
Let $\{P_1,\dots,P_n\}\subset\mathcal P(\Hcal_E)$ be a Kochen–Specker
uncolourable set whose projectors sum to the identity operator on $\Hcal_E$, i.e.\ there exists \emph{no} map
$v:\mathcal P(\Hcal_E)\!\to\!\{0,1\}$ satisfying
$\sum_{i=1}^{k}v(\Pi_i)=1$ for every projector decomposition
$\sum_{i=1}^{k}\Pi_i=\mathbb{I}$ that contains only elements of the set.
Let $\mathcal I$ be any instrument whose Kraus operators are polynomials in the $P_j$ and whose action is confined to the support $\Hcal_E$.
Then no quantum Turing machine~$Q$ can output, for every action
sequence $a_{1:m}$ an agent might perform, a commuting family of
projectors $\{Q_{a_{1:m}}(e_{1:m})\}$ such that
\begin{equation}\label{eq:perfect_prediction}
   \forall a_{1:m},\;
   \forall e_{1:m}\in\Gamma^{m}_{\mathrm{obs}}:\quad
   \Pr_{\rm env}(e_{1:m}\Vert a_{1:m})
   \;=\; 
   \Tr \!\bigl[
      Q_{a_{1:m}}(e_{1:m})\,
      \rho^\star_E
   \bigr],
\end{equation}
with $\rho^\star_E$ the true environment state prepared before cycle~$1$.
\end{corollary}

\begin{proof}
Fix one action string $a_{1:m}^\dagger$ that instructs the agent
to measure, at the final cycle, \emph{every} projector in the KS-set.
Because the $Q$TM outputs commuting projectors
$\bigl\{Q_{a_{1:m}^\dagger}(e_{1:m})\bigr\}$, for each $j$ there is a
classical random variable
\(
   X_j := v_Q(P_j)
   := e_{1:m}(P_j)\in\{0,1\},
\)
namely the indicator that the outcome sequence has eigenvalue $1$ for
$P_j$.  
Equation~\eqref{eq:perfect_prediction} asserts that these random
variables reproduce the Born probabilities of the true state, hence the
map
\(
   v_Q:\;P_j\longmapsto X_j(\omega)
\)
is \emph{non-contextual}: it assigns $0$ or $1$ to $P_j$ without
reference to the context in which $P_j$ is measured. By construction
$\sum_{j\,\in\,\mathcal C}P_j=\mathbb{I}$ for every context
$\mathcal C\subset\{1,\dots,n\}$, whence
\(
   \sum_{j\,\in\,\mathcal C} v_Q(P_j) = 1
\)
holds almost surely.
Thus $v_Q$ is a \emph{global}
transformation of the projector into $\{0,1\}$,
contradicting KS uncolourability.
Therefore a perfect, non-disturbing predictor $Q$ cannot exist.
\end{proof}
\vspace{-0.5em}
Contextuality problematises the straightforward analogy with classical history in AIXI albeit subject to the extent to which environments actually do manifest uncolourable equivalents. Because a universal classical history $h_{<t}$ cannot determine
simultaneously the outcomes of all future instruments, any Bayesian
update rule that conditions the QSI posterior on $h_{<t}$ \emph{alone}
is necessarily incomplete; the posterior must be refined by the
entire \textit{future} instrument schedule.  Consequently the
martingale proof of Theorem~\ref{thm:qsi-convergence} requires an
adapted filtration that records measurement contexts.
\vspace{-1em}

\subsubsection{No-cloning.}  
      Evaluating the quantum likelihood
      $\Pr_{Q}(e_{1:m}\Vert a_{1:m})$ requires \emph{fresh} copies of
      the pre-measurement state $\rho_{E}^{Q}(a_{1:m})$, but the
      no-cloning theorem forbids their duplication from a single run.
      Hence each Bayesian update consumes the very evidence it needs
      for validation; the sample complexity of learning scales with
      the number of distinct instruments explored, whereas in the
      classical case the same trajectory can be replayed arbitrarily
      often. Two consequences of the above are that: (a) convergence proofs must be context-sensitive; a single
      universal mixture cannot assign fixed probabilities to all
      quantum experiments simultaneously; and (b)  even under ideal identifiability the learning rate is limited by state-preparation resources giving rise to sample complexity consequences for any QAIXI agent. See Appendix \ref{app:no-cloning-sample-complexity}.
      \vspace{-1em}

\section{Conclusion and Open Questions}
\vspace{-0.5em}
We have laid out an analysis of and quantum information processing-based mathematical framework for Quantum AIXI, a universal intelligent agent designed to operate within a quantum mechanical universe. We have shown, using our channel and register-based model of agent/environment interaction in the quantum setting, how quantum analogues of universal intelligence components can be constructed. However there are significant limitations to QAIXI. Firstly, it is not practical. QEC requires millions of physical qubits for even modest logical qubit counts. Maintaining extended coherence for QAIXI decision cycles remains far beyond existing technology. Key challenges include the incomputability of quantum Kolmogorov complexity, the resource overhead required for quantum state preparation, simulation and calculating probabilities. Moreover, exactly when QAIXI would provide a quantum advantage is an open question (see App. \ref{app:quantum-advantage}).
As interest in both computational agency and quantum information processing grows, understanding the capabilities and constraints of integrating quantum technologies into AGI systems will become of increasing focus. Further research directions include:
\begin{enumerate}
    \item \textit{Computable Approximations}: Can meaningful, computable (or efficiently quantum-computable) approximations to QAIXI be developed e.g. via classical shadows \cite{huang2024learning,huang2020predicting,huang2022quantum}? What restrictions on the class of quantum environments $\Mcal_{Q,sol}$ or the definition of $K_Q$ would make this possible?
    \item \textit{Convergence Rates and Bounds}: Rigorously establishing convergence theorems for QSI and deriving tight bounds on convergence rates, taking into account quantum information-theoretic limits (e.g., Holevo information, quantum data processing inequalities), in light of learnability and tomography \cite{aaronson2007learnability} e.g. How does quantum measurement back-action affect learning speed?
    \item \textit{Quantum Resources}: How do quantum resources like entanglement, if available to the agent for its internal processing or for interacting with the environment, affect the performance or complexity of QAIXI? How feasible is coherent learning \cite{lupu2024qubits,pang2024information} or coherent histories for QAIXI? 
    \item \textit{Impact of Quantum Interpretations}: While this paper focused on the standard (Copenhagen-like) measurement formalism, how would adopting alternative interpretations (e.g., Everettian Many-Worlds, Bohmian Mechanics \cite{Bohm1952}, QBism \cite{fuchs2014introduction}) alter the definition of QAIXI's optimality, its policy, or its perceived complexity? For example, in an Everettian setting, does QAIXI optimise expected reward across trajectories or branches \cite{Deutsch1999,Wallace2012}?
\end{enumerate}

%
%
\bibliographystyle{splncs04}
\bibliography{refs-quant-ontology,refs-new,refs-agi,bibliograph1,refs-agi-classical} 

\begin{thebibliography}{10}
\providecommand{\url}[1]{\texttt{#1}}
\providecommand{\urlprefix}{URL }
\providecommand{\doi}[1]{https://doi.org/#1}

\bibitem{aaronson2007learnability}
Aaronson, S.: The learnability of quantum states. Proceedings of the Royal Society A: Mathematical, Physical and Engineering Sciences  \textbf{463}(2088),  3089--3114 (2007)

\bibitem{aaronson2013quantum}
Aaronson, S.: Quantum computing since Democritus. Cambridge University Press (2013)

\bibitem{aaronson2014quantum}
Aaronson, S.: Quantum machine learning algorithms: Read the fine print. Nature Physics p.~5 (2014)

\bibitem{aaronson2018shadow}
Aaronson, S.: Shadow tomography of quantum states. In: Proceedings of the 50th annual ACM SIGACT symposium on theory of computing. pp. 325--338 (2018)

\bibitem{AaronsonArkhipov2011}
Aaronson, S., Arkhipov, A.: The computational complexity of linear optics. In: Proceedings of the forty-third annual ACM symposium on Theory of computing. pp. 333--342 (2011)

\bibitem{bell_speakable_2004}
Bell, J.: Speakable and {Unspeakable} in {Quantum} {Mechanics}. {Cambridge} {University} {Press}, Cambridge, 2nd edn. (2004)

\bibitem{bennett2023b}
Bennett, M.T.: The optimal choice of hypothesis is the weakest, not the shortest. In: Artificial General Intelligence. Springer Nature (2023)

\bibitem{bennett2024appendices}
Bennett, M.T.: Technical appendices (2024). \doi{10.5281/zenodo.7641741}, \url{https://github.com/ViscousLemming/Technical-Appendices}

\bibitem{bennettmaruyama2022b}
Bennett, M.T., Maruyama, Y.: The artificial scientist: Logicist, emergentist, and universalist approaches to artificial general intelligence. In: Artificial General Intelligence. Springer (2022)

\bibitem{berthiaume2001quantum}
Berthiaume, A., Van~Dam, W., Laplante, S.: Quantum kolmogorov complexity. Journal of Computer and System Sciences  \textbf{63}(2),  201--221 (2001)

\bibitem{BiamonteEtAl2017}
Biamonte, J., Wittek, P., Pancotti, N., Rebentrost, P., Wiebe, N., Lloyd, S.: Quantum machine learning. Nature  \textbf{549}(7671),  195--202 (2017)

\bibitem{Bohm1952}
Bohm, D.: A suggested interpretation of the quantum theory in terms of "hidden" variables. i and ii. Physical Review  \textbf{85}(2),  166--193 (1952)

\bibitem{bostanci2022quantum}
Bostanci, J., Watrous, J.: Quantum game theory and the complexity of approximating quantum nash equilibria. Quantum  \textbf{6}, ~882 (2022)

\bibitem{brukner2014quantum}
Brukner, {\v{C}}.: Quantum causality. Nature Physics  \textbf{10}(4),  259--263 (2014)

\bibitem{catt2020}
Catt, E., Hutter, M.: A gentle introduction to quantum computing algorithms with applications to universal prediction (2020)

\bibitem{cerezo_challenges_2022}
Cerezo, M., Verdon, G., Huang, H.Y., Cincio, L., Coles, P.J.: Challenges and opportunities in quantum machine learning. Nature Computational Science  (2022)

\bibitem{csiszar2011information}
Csisz{\'a}r, I., K{\"o}rner, J.: Information theory: coding theorems for discrete memoryless systems. Cambridge University Press (2011)

\bibitem{deutsch1985}
Deutsch, D.: Quantum theory, the church-turing principle and the universal quantum computer. Proceedings of the Royal Society of London. A. Mathematical and Physical Sciences  \textbf{400}(1818),  97--117 (1985)

\bibitem{Deutsch1999}
Deutsch, D.: Quantum theory of probability and decisions. Proceedings of the Royal Society of London. Series A: Mathematical, Physical and Engineering Sciences  \textbf{455}(1988),  3129--3137 (1999)

\bibitem{DongEtAl2006}
Dong, D., Chen, C., Chen, H., Tarn, T.J.: Quantum reinforcement learning. IEEE Transactions on Systems, Man, and Cybernetics, Part B (Cybernetics)  \textbf{38}(5),  1207--1220 (2008)

\bibitem{DunjkoBriegel2018}
Dunjko, V., Briegel, H.J.: Machine learning \& artificial intelligence in the quantum domain: a review of recent progress. Reports on Progress in Physics  \textbf{81}(7),  074001 (2018)

\bibitem{ebtekar2025foundations}
Ebtekar, A., Hutter, M.: Foundations of algorithmic thermodynamics. Physical Review E  \textbf{111}(1),  014118 (2025)

\bibitem{fallenstein2015reflective}
Fallenstein, B., Soares, N., Taylor, J.: Reflective variants of solomonoff induction and aixi. In: International Conference on Artificial General Intelligence. pp. 60--69. Springer (2015)

\bibitem{fang2020chain}
Fang, K., Fawzi, O., Renner, R., Sutter, D.: Chain rule for the quantum relative entropy. Physical review letters  \textbf{124}(10),  100501 (2020)

\bibitem{fuchs2014introduction}
Fuchs, C.A., Mermin, N.D., Schack, R.: An introduction to qbism with an application to the locality of quantum mechanics. American Journal of Physics  \textbf{82}(8),  749--754 (2014)

\bibitem{goertzel2021}
Goertzel, B.: The general theory of general intelligence: A pragmatic patternist perspective. Tech. rep., Singularity Net (2021)

\bibitem{goertzel2023}
Goertzel, B., et~al.: Opencog hyperon: A framework for agi at the human level and beyond. Tech. rep., OpenCog (2023)

\bibitem{grunwald2008algorithmic}
Gr{\"u}nwald, P.D., Vit{\'a}nyi, P., et~al.: Algorithmic information theory. Handbook of the Philosophy of Information pp. 281--320 (2008)

\bibitem{guo2008survey}
Guo, H., Zhang, J., Koehler, G.J.: A survey of quantum games. Decision Support Systems  \textbf{46}(1),  318--332 (2008)

\bibitem{gutoski2007toward}
Gutoski, G., Watrous, J.: Toward a general theory of quantum games. In: Proceedings of the thirty-ninth annual ACM symposium on Theory of computing. pp. 565--574 (2007)

\bibitem{haapasalo2021choi}
Haapasalo, E.: The choi--jamio{\l}kowski isomorphism and covariant quantum channels. Quantum Studies: Mathematics and Foundations  \textbf{8}(3),  351--373 (2021)

\bibitem{hirota2020application}
Hirota, O.: Application of quantum pinsker inequality to quantum communications. arXiv preprint arXiv:2005.04553  (2020)

\bibitem{huang2024learning}
Huang, H.Y.: Learning in the Quantum Universe. California Institute of Technology (2024)

\bibitem{huang2022quantum}
Huang, H.Y., Broughton, M., Cotler, J., Chen, S., Li, J., Mohseni, M., Neven, H., Babbush, R., Kueng, R., Preskill, J., et~al.: Quantum advantage in learning from experiments. Science  \textbf{376}(6598),  1182--1186 (2022)

\bibitem{huang2020predicting}
Huang, H.Y., Kueng, R., Preskill, J.: Predicting many properties of a quantum system from very few measurements. Nature Physics  \textbf{16}(10),  1050--1057 (2020)

\bibitem{hutter2004universal}
Hutter, M.: Universal artificial intelligence: Sequential decisions based on algorithmic probability. Springer Science \& Business Media (2004)

\bibitem{hutter2007}
Hutter, M.: Universal Algorithmic Intelligence: A Mathematical Top{\textrightarrow}Down Approach, pp. 227--290. Springer Berlin Heidelberg, Berlin, Heidelberg (2007)

\bibitem{hutter2010}
Hutter, M.: Universal Artificial Intelligence: Sequential Decisions Based on Algorithmic Probability. Springer, Heidelberg (2010)

\bibitem{jerbi_parametrized_2021}
Jerbi, S., Gyurik, C., Marshall, S., Briegel, H., Dunjko, V.: Parametrized quantum policies for reinforcement learning. Advances in Neural Information Processing Systems  \textbf{34},  28362--28375 (2021)

\bibitem{knill1995approximation}
Knill, E.: Approximation by quantum circuits. arXiv preprint quant-ph/9508006  (1995)

\bibitem{kochenSpecker1967}
Kochen, S., Specker, E.P.: The problem of hidden variables in quantum mechanics. Journal of Mathematics and Mechanics  \textbf{17}(1),  59--87 (1967)

\bibitem{leike2015}
Leike, J., Hutter, M.: Bad universal priors and notions of optimality. COLT  (2015)

\bibitem{li2019introduction}
Li, M., Vit{\'a}nyi, P.M.B.: An Introduction to {K}olmogorov Complexity and Its Applications. Springer, 4 edn. (2019)

\bibitem{liu_rigorous_2021}
Liu, Y., Arunachalam, S., Temme, K.: A rigorous and robust quantum speed-up in supervised machine learning. Nature Physics pp.~1--5 (2021)

\bibitem{lupu2024qubits}
Lupu-Gladstein, N., Brodutch, A., Ferretti, H., Tham, W.K., Pang, A.O.T., Bonsma-Fisher, K., Steinberg, A.M.: Do qubits dream of entangled sheep? quantum measurement without classical output. New Journal of Physics  \textbf{26}(5),  053029 (2024)

\bibitem{mcmillen2024}
McMillen, P., Levin, M.: Collective intelligence: A unifying concept for integrating biology across scales and substrates. Communications Biology  \textbf{7}(1), ~378 (Mar 2024)

\bibitem{meyer2022survey}
Meyer, N., Ufrecht, C., Periyasamy, M., Scherer, D.D., Plinge, A., Mutschler, C.: A survey on quantum reinforcement learning. arXiv preprint arXiv:2211.03464  (2022)

\bibitem{NielsenChuang2010}
Nielsen, M.A., Chuang, I.L.: Quantum Computation and Quantum Information. Cambridge University Press, 10th anniversary edn. (2010)

\bibitem{ozkural2015ultimate}
{\"O}zkural, E.: Ultimate intelligence part ii: Physical measure and complexity of intelligence. arXiv preprint arXiv:1504.03303  (2015)

\bibitem{ozkural2016ultimate}
{\"O}zkural, E.: Ultimate intelligence part ii: physical complexity and limits of inductive inference systems. In: Artificial General Intelligence: 9th International Conference, AGI 2016, New York, NY, USA, July 16-19, 2016, Proceedings 9. pp. 33--42. Springer (2016)

\bibitem{pang2024information}
Pang, A.O., Lupu-Gladstein, N., Yilmaz, Y.B., Brodutch, A., Steinberg, A.M.: Information gain and measurement disturbance for quantum agents. arXiv preprint arXiv:2402.08060  (2024)

\bibitem{peres1991two}
Peres, A.: Two simple proofs of the kochen-specker theorem. Journal of Physics A: Mathematical and General  \textbf{24}(4), ~L175 (1991)

\bibitem{perrier2024quantum}
Perrier, E.: {Q}uantum {G}eometric {M}achine {L}earning. arXiv:2409.04955  (2024)

\bibitem{perrier2025qaxirefs}
Perrier, E.: Quantum {AIXI} - technical appendices (2025). \doi{10.5281/zenodo.15645658}, \url{https://zenodo.org/records/15645658}

\bibitem{sarkar2022qksa}
Sarkar, A., Al-Ars, Z., Bertels, K.: Qksa: Quantum knowledge seeking agent. In: International Conference on Artificial General Intelligence. pp. 384--393. Springer (2022)

\bibitem{schuld_evaluating_2019}
Schuld, M., Bergholm, V., Gogolin, C., Izaac, J., Killoran, N.: Evaluating analytic gradients on quantum hardware. Physical Review A  \textbf{99}(3) (mar 2019)

\bibitem{SchuldPetruccione2018}
Schuld, M., Petruccione, F.: Supervised Learning with Quantum Computers. Springer (2018)

\bibitem{schuld_machine_2021}
Schuld, M., Petruccione, F.: Machine {Learning} with {Quantum} {Computers}. Springer (2021)

\bibitem{soares2015two}
Soares, N., Fallenstein, B.: Two attempts to formalize counterpossible reasoning in deterministic settings. In: Artificial General Intelligence: 8th International Conference, AGI 2015, AGI 2015, Berlin, Germany, July 22-25, 2015, Proceedings 8. pp. 156--165. Springer (2015)

\bibitem{solomonoff1964}
Solomonoff, R.J.: A {F}ormal {T}heory of {I}nductive {I}nference. {P}art {I} \& {II}. Information and Control  \textbf{7}(1--2),  1--22, 224--254 (1964)

\bibitem{sole2019}
Solé, R., Moses, M., Forrest, S.: Liquid brains, solid brains. Philosophical Transactions of the Royal Society B: Biological Sciences  \textbf{374}(1774),  20190040 (2019)

\bibitem{thorisson2016artificial}
Th{\'o}risson, K.R., Bieger, J., Thorarensen, T., Sigur{\dh}ard{\'o}ttir, J.S., Steunebrink, B.R.: Why artificial intelligence needs a task theory: and what it might look like. In: Artificial General Intelligence: 9th International Conference, AGI 2016, New York, NY, USA, July 16-19, 2016, Proceedings 9. pp. 118--128. Springer (2016)

\bibitem{vidal2003efficient}
Vidal, G.: Efficient classical simulation of slightly entangled quantum computations. Physical review letters  \textbf{91}(14),  147902 (2003)

\bibitem{vitanyi2002quantum}
Vit{\'a}nyi, P.M.: Quantum kolmogorov complexity based on classical descriptions. IEEE Transactions on Information Theory  \textbf{47}(6),  2464--2479 (2002)

\bibitem{Wallace2012}
Wallace, D.: The Emergent Multiverse: Quantum Theory according to the Everett Interpretation. Oxford University Press (2012)

\bibitem{wang2015assumptions}
Wang, P., Hammer, P.: Assumptions of decision-making models in agi. In: Artificial General Intelligence: 8th International Conference, AGI 2015, AGI 2015, Berlin, Germany, July 22-25, 2015, Proceedings 8. pp. 197--207. Springer (2015)

\bibitem{watrous_theory_2018}
Watrous, J.: The {Theory} of {Quantum} {Information}. {Cambridge} {University} {Press} (2018)

\bibitem{wiseman2016causarum}
Wiseman, H.M., Cavalcanti, E.G.: Causarum investigatio and the two bell’s theorems of john bell. In: Quantum [Un] Speakables II: Half a Century of Bell's Theorem, pp. 119--142. Springer (2016)

\end{thebibliography}

\newpage
\appendix

\section{Entanglement in the Interaction Loop}
\label{app:entangled-loop}
In \S\ref{sec:qagi-models}, the post-measurement joint state was written under the
separability assumption $\rho_{AE}^{(t-1)}=\rho_{A}^{(t-1)}\!\otimes\rho_{E}^{(t-1)}$,
leading to
\(
   \rho_{AE}^{(t)}=\rho_{A}^{(t-1)}\!\otimes\rho_{E'}^{(t)}
\).
When the agent and environment registers are initially entangled
this factorisation is instead as follows. The outcome probability is given by:
\begin{align}
    \Pr(o_t=k\,|\,a_t,\rho_{AE}^{(t-1)})
          \;=\;
          \operatorname{Tr}\!\Bigl[
             \bigl(\operatorname{id}_{A}\!\otimes\!\mathcal E^{a_t}_{k}\bigr)
             \rho_{AE}^{(t-1)}
          \Bigr] \label{eq:app:entanglement}
\end{align}
with conditional post-measurement state:
\begin{align*}
    \rho_{AE}^{(t)}
          \;=\;
          \frac{
            (\operatorname{id}_{A}\!\otimes\!\mathcal E^{a_t}_{k})
            \bigl(\rho_{AE}^{(t-1)}\bigr)
          }{\Pr(k)}
\end{align*}
and reduced environmental state:
\begin{align*}
    \rho_{E'}^{(t)}
            =\operatorname{Tr}_{A}\rho_{AE}^{(t)}
            =\frac{\mathcal E^{a_t}_{k}\!\bigl(
                       \operatorname{Tr}_{A}\rho_{AE}^{(t-1)}
                     \bigr)}
                   {\Pr(k)}.
\end{align*}
Only when $\rho_{AE}^{(t-1)}$ is separable (or when the instrument
completely decoheres the measured subsystem) does the joint state
factorise into $\rho_{A}^{(t-1)}\!\otimes\rho_{E'}^{(t)}$. Note the superoperator $\operatorname{id}_{A}\!\otimes\!\mathcal E^{a_t}_{k}$ acts on $\rho_{AE}$ regardless of whether a product state or not.  The subsequent internal QAIXI update
$\mathsf U_{\!\mathrm{int}}:\rho_{A}^{(t-1)}\!\mapsto\rho_{A}^{(t)}$
now acts on 
\(
  \rho_{A}^{(t)}=\operatorname{Tr}_{E}\rho_{AE}^{(t)},
\)
which already incorporates any entanglement-induced disturbance. This general form of interaction loop is valid for all
pre-measurement states, entangled or not, and serves as the reference implementation for any rigorous analysis of QAIXI in the fully quantum
regime. Note that in the main text we have adopted the separability assumption in (\S\ref{sec:qagi-models}) i.e. assuming a product–state
form~$\rho_A^{(t-1)}\!\otimes\rho_E^{(t-1)}$ for derivations of the quantum Bellman equation and other discussion. However, we have only qualitatively sketched some of the consequences of QAIXI-environment entanglement in the remainder of the paper due to space constraints. A more detailed technical discussion of the same is the topic of forthcoming work.

\section{Quantum Advantage for QAIXI}
\label{app:quantum-advantage}
A fundamental question for any quantum extension of classical AIXI is: under what circumstances does quantum processing provide genuine computational advantages? We examine this question below in the case of specific environment classes where classical simulation becomes intractable.

\subsection{Boson Sampling Environments}
A common reference for assessments of quantum computational advantage over classical protocols is that of boson sampling \cite{AaronsonArkhipov2011}, a non-universal photonic quantum protocol that injects indistinguishable single photons into a fixed linear interferometer and samples output patterns whose probabilities are governed by matrix permanents which is considered exponentially hard for classical computers. In principle boson-sampling would, as with `non-agentic' quantum systems, propose a candidate test of quantum advantage (albeit with limited practical scope). To see how it would work, let $\mathcal{B}_n$ denote the class of quantum environments that embed boson sampling processes as follows. Each environment $Q \in \mathcal{B}_n$ maintains an $n$-mode optical interferometer described by a random unitary matrix $U \in \mathrm{U}(n)$ drawn from the Haar measure. On each cycle, $Q$ accepts a classical action $a_t$ specifying input photon configuration. It then evolves the state via $U$, and outputs measurement outcomes from a fixed detection scheme. Suppose there exists a family $\{Q_n\} \subset \mathcal{B}_n$ such that (a) generating one percept from $Q_n$ requires $\mathrm{poly}(n)$ quantum operations; and (b) computing the output distribution to total variation distance $< 2^{-n}$ is $\#\mathrm{P}$-hard. In such a case, any classical AIXI with total running time fails to approximate the optimal value unless the polynomial hierarchy collapses.
\\
\\
This can be understood by considering how AIXI's value function requires accurate estimation of environment transition probabilities. For $Q_n \in \mathcal{B}_n$, this reduces to sampling from the boson sampling distribution. The original boson sampling conjecture involves classical simulation which, to be within the required accuracy, would yield an algorithm for computing permanents of Gaussian random matrices—a problem conjectured to be $\#\mathrm{P}$-complete which would collapse the hierarchy, contradicting standard complexity-theoretic assumptions. In terms of value functions, let $\hat{V}^{\mathrm{classical}}$ denote the value function computed by any polynomial-time classical agent. If $|\hat{V}^{\mathrm{classical}} - V^*_{Q_n}| < 1/\mathrm{poly}(n)$, a QAIXI agent would have to simulate the boson-sampling output well enough to capture its fine-grained probabilities. Aaronson et al. showed that any classical routine able to do that would also let us efficiently approximate the permanent of a random matrix—something believed so hard that it would topple the entire polynomial-hierarchy ladder of complexity classes \cite{AaronsonArkhipov2011}. In this way, boson sampling provides an in principle litmus test for quantum advantage, albeit with limited applicability to specific tasks AIXI or QAIXI would undertake. The boson sampling example illustrates how QAIXI's computational advantage over classical AIXI is contingent on environments involving problems which are classically hard (see discussion in \cite{AaronsonArkhipov2011}).

\subsection{Structured Environment Classes}

The example above raises the obvious question of when do quantum environments admit efficient classical approximation? We draw upon results in quantum information to illustrate when this may be true in the case of weakly entangled systems. In quantum formalism, an environment $Q$ is called a \textit{Matrix Product Environment} (MPE) with bond dimension $\chi$ if its state evolution can be represented as:
\begin{equation}
\rho_E^{(t)} = \mathrm{Tr}_{\mathrm{aux}}\left[\bigotimes_{i=1}^{n} A_i^{[t]}(\chi)\right]
\end{equation}
where each $A_i^{[t]}(\chi)$ is a $\chi \times \chi$ matrix depending on the history up to time $t$ (see also \cite{vidal2003efficient}). The question then becomes how classically learnable is an MPE? Let $\mathcal{E}_{\chi}$ denote the class of all MPEs with bond dimension $\chi = O(1)$. Then there exists a classical polynomial-time agent that $\varepsilon$-approximates QAIXI's value function on any $Q \in \mathcal{E}_{\chi}$ with sample complexity $O(\mathrm{poly}(n, \varepsilon^{-1}))$ assuming local measurements . The bond dimension of an MPE is the size of tensors used to represent the MPE and is a figure that relates to the complexity of the representation and the extent of its entanglement. The key observation is that MPEs with constant bond dimension admit efficient tomography using classical shadow tomography \cite{aaronson2018shadow,huang2020predicting} which can, in certain cases, reduce the observables required to calculate the expectation values of exponentially many observables of an unknown quantum state to only logarithmic number of randomised measurements of that state. We now turn to general information-theoretic limits.

\subsection{Quantum learning and classical shadows}
The focus of our examination of QAIXI above is upon a quantum agent which may learn to act optimally in a universe that is itself quantum mechanical. However, QAIXI depends upon an uncomputable mixture $\Xi_Q$ while the sample-complexity consequences of estimating $\Xi_Q$ by full state or process tomography incurs an exponential overhead that would undermine any potential quantum advantage. One prospective solution lies in exploring the role that novel techniques in classical shadow tomography \cite{huang2020predicting,huang2022quantum,huang2024learning} may offer . Shadow-based protocols exploit randomised Clifford measurements followed by a classical post-processing map $\mathsf{S}$ that builds a shadow $\hat\rho$ of the unknown state from a compressed set of single-copy outcomes. We explore the adaption of QAIXI using classical shadows in ongoing work.


\section{QSI Detail and Limitations}
\label{app:qsi-convergence-elaboration}
Here we discuss conditions on QSI (Theorem~\ref{thm:qsi-convergence}) in particular some of the challenges facing proving convergence.

\subsection{Convergence Conditions}
Prior to measurement, each state is given a weight proportional to its description length. Recall from Theorem~\ref{thm:qsi-convergence} that under assumptions \emph{(C1)}–\emph{(C3)}, the posterior semi-density operator $\Xi_Q^{(t)}$ (Eq.~\eqref{eq:posteriorupdate}) satisfies Eq. \ref{eqn:QSI:posteriordensity}:
\begin{align}
    \mathbb E_{Q^{\star}}\!\bigl[
     D\!\bigl(
        \rho_{E}^{\star}(a_{1:t})
        \,\big\Vert\,
        \Xi_{Q}^{(t)}(a_{1:t})
     \bigr)
  \bigr] = \Tr\Bigl[
      \rho_{E}^{\star}(a_{1:t})
      \Bigl(
         \ln \rho_{E}^{\star}(a_{1:t})
         -
         \ln \Xi_{Q}^{(t)}(a_{1:t})
      \Bigr)
  \Bigr]
  \;\le\;
  \frac{K_{Q}(Q^{\star})\ln 2+\ln(1+g)}{\!\!t}.
\end{align}
To show this, recall that each density operator $\rho_E^Q$ is associated with a quantum state $Q$ where $Q^*$ denotes the true environmental state. $\Xi_{Q}^{(t)}$ describes the agent's best assessment as to the true state of the environment $\rho_E^*$ at time $t$ (we leave the time dependence of $\rho$ understood).  The QAIXI belief state is a weighting using such descriptions over each possible state density operator:
     \begin{align*}
         \Xi_Q = \omega_{Q^*} \rho_E^* + \sum_{Q \neq Q^*} \omega_Q \rho_E^Q.
     \end{align*}  
The $\omega_Q$ terms reflect the descriptions of state $Q$. The shortest description of $Q$ is given by $K_Q(Q)$ where simpler descriptions get more weight via $\omega_Q = 2^{-K_Q(Q)}$. In classical minimum description length (MDL) formalism, the code length of a data sequence is a combination of the model complexity (description length) and likelihood of data under the chosen model (a dilution term). 
The $D\!\bigl(
        \rho_{E}^{\star}(a_{1:t})
        \,\big\Vert\,
        \Xi_{Q}^{(t)}(a_{1:t})
     \bigr)$ term denotes the relative entropy and is reflective of the additional measurements required to minimise the estimate between $\Xi$ and $\rho_E^*$. To see this, define:
     \begin{align}
         Z = \omega_{Q^*} + \sum_{Q \neq Q^*} \omega_Q = \omega_{Q^*}(1+g) \label{app:QSI:Z}
     \end{align}
     where:
     \begin{align}
         g = \sum_{Q\neq Q^*} \frac{\omega_Q}{\omega_{Q^*}} = \sum_{Q \neq Q^*} 2^{-(K_Q(Q) - K_Q(Q^*))} \label{app:QSI:g}
     \end{align}
     $\Xi_Q$ is not yet normalised and $\rho_E$ remains a semi-density operator as  $\Tr \rho_E \leq 1$ i.e. as with the classical case $Z$ is sub-normalised e.g. $Z = \sum_Q \omega_Q \leq 1$ for prefix-free codes with binary word-lengths (see \cite{grunwald2008algorithmic}). To apply the martingale and
   chain-rule identities we require density (not semi-density) operators. To do so, we note the correct model's share of the total description mass (the effective prior mass of the true environment inside the normalised
mixture) is:
     \begin{align*}
         \frac{\omega_{Q^*}}{Z} &= \frac{1}{1+g} = c_{max} \\
         Z &= \omega_{Q^*}(1 + g)
     \end{align*}
      The normalised prior at time $t=0$ is then given by:
     \begin{align}
         \Xi_Q = \frac{\Xi_Q^{(0)}}{Z} = \frac{\rho^*_E}{(1+g)}
     \end{align} 
    The form of the bound $D_0$ depends on our choice of weighting $c$ for $\rho_E^*$ where $0 < c \leq c_{max}$. Choosing $c=c_{max}$ gives:
    \begin{align}
        \Xi_Q^{(0)} = c\rho_E^* = \frac{1}{1+g} \rho_E^*
    \end{align}
     with relative entropy given by:
\begin{align*}
         D_0 &= \Tr[\rho_E(\ln \rho^*_E - \ln \Xi_Q)]\\
         &\leq \Tr[\rho_E(\ln \rho^*_E - (\ln \rho_E^* - \ln(1+g)))]\\
         &= \ln(1+g)
     \end{align*}
      assuming $\Tr(\rho_E^Q)=1$ for all states $Q$. This form includes only the dilution term. In the event of perfect information (so $g=0$) it reduces to the complexity $K_Q(Q^*) \ln 2$. To assume the more familiar MDL form (which includes both complexity and dilution terms) we can scale $c$ (while remaining within the bounds above) e.g. $c=\omega_{Q^*}^2/Z = \omega_{Q^*}/(1+g)$. Doing so gives:
      \begin{align}
         D_0 &= \Tr[\rho_E(\ln \rho^*_E - \ln \Xi_Q)]\\
         &\leq \Tr[\rho_E(\ln \rho^*_E - (\ln \omega_{Q^*}  - \ln(1+g) + \ln \rho_E^*))]\\
         &= \underbrace{K_Q(Q^*) \ln 2}_{\text{model complexity}} + \underbrace{\ln(1+g)}_{\text{dilution}} \label{eq:1}
     \end{align}
At each cycle $s$ (action and observation), the agent’s
state update is a CPTP map
$\mathcal M_{s}:\rho\mapsto\rho(a_{1:s})$ such that: 
\begin{align}
    \Xi_{Q}^{(s)}
   \;=\;
   \frac{\mathcal M_{s}\!\bigl(\Xi_{Q}^{(s-1)}\bigr)}
        {\operatorname{Tr}\!\mathcal M_{s}\!\bigl(\Xi_{Q}^{(s-1)}\bigr)}
   \quad\text{and}\quad
\rho^{\star}_{E}(a_{1:s})
   \;=\;
   \frac{\mathcal M_{s}\!\bigl(\rho^{\star}_{E}(a_{1:s-1})\bigr)}
        {\operatorname{Tr}\!\mathcal M_{s}\!\bigl(\rho^{\star}_{E}(a_{1:s-1})\bigr)} .
\end{align}
Using the chain rule for quantum relative entropy
\cite{fang2020chain} we obtain:
\begin{equation}
D_{s-1}
   =
   \operatorname*{\mathbb E}_{Q^{\star}}
   \bigl[
      D\bigl(\rho^{\star}_{E}(a_{1:s-1})\Vert\Xi_{Q}^{(s-1)}\bigr)
   \bigr]
   =
   \operatorname*{\mathbb E}_{Q^{\star}}
   \sum_{k=1}^{s-1}
   D\!\bigl(
        \rho^{\star}_{E}(a_{k}\!\mid a_{1:k-1})
        \,\big\Vert\,
        \Xi_{Q}^{(k)}(a_{k}\!\mid a_{1:k-1})
   \bigr).
\label{eq:2}
\end{equation}
Applying \eqref{eq:2} with $s=t+1$ and noting that each conditional divergence
is non-negative gives:
\begin{equation}
\operatorname*{\mathbb E}_{Q^{\star}}
   D\!\bigl(
        \rho^{\star}_{E}(a_{1:t})
        \,\Vert\,
        \Xi_{Q}^{(t)}(a_{1:t})
      \bigr)
   \;\le\;
   D_{0}.
\tag{3}\label{eq:3}
\end{equation}
Combining \ref{eq:1} and \eqref{eq:3} and dividing by $t$:
\begin{align}
    \mathbb E_{Q^{\star}}
   D\!\bigl(
        \rho^{\star}_{E}(a_{1:t})
        \,\Vert\,
        \Xi_{Q}^{(t)}(a_{1:t})
      \bigr)
   \;\le\;
   \frac{K_{Q}(Q^{\star})\ln 2+\ln(1+g)}{t}
\end{align}
to obtain Eq. \ref{eqn:QSI:posteriordensity}. In our sketch above, the transition from the relative entropy convergence to a convergence rate for distinguishability is achieved using the quantum Pinsker inequality \cite{watrous_theory_2018}, quantifying how the QAIXI agent's belief $\Xi_{Q}^{(t)}(a_{1:t})$ converges to the true state $\rho_{E}^{\star}(a_{1:t})$. In an information-theoretic sense, $\frac{1}{2}\Vert \rho - \sigma \Vert_1$ represents the maximum probability with which one can distinguish between two quantum states $\rho$ and $\sigma$ using an optimal measurement. The quantum Pinsker inequality (also known as the Csiszár–Kullback–Pinsker inequality for quantum states) provides that for any two density operators $\rho$ and $\sigma$:
\[
\tfrac12\Vert\rho-\sigma\Vert_{1}
  \;\;\le\;\;
  \sqrt{\tfrac12\,D(\rho\Vert\sigma)}.
\]
We use the property that for non-negative random variables, if $A \le B$, then $\mathbb{E}[A] \le \mathbb{E}[B]$ and Jensen's inequality such that $\mathbb{E}[\sqrt{X}] \le \sqrt{\mathbb{E}[X]}$ such that:
\begin{equation}
  \mathbb E_{Q^{\star}}\!
  \bigl[\tfrac12\lVert\rho-\sigma\rVert_{1}\bigr]
  \;\le\;
  \sqrt{\tfrac12\;
        \mathbb E_{Q^{\star}}\!\bigl[D(\rho\Vert\sigma)\bigr]}.
\label{eq:pinsker_expected_value_form}
\end{equation}
Substituting into equation  \ref{eq:pinsker_expected_value_form} with $\rho=\rho_{E}^{\star}(a_{1:t})$ and $\sigma=\Xi_{Q}^{(t)}(a_{1:t})$:
\begin{align*}
\mathbb E_{Q^{\star}}\!
  \Bigl[
     \tfrac12
     \bigl\lVert
       \rho_{E}^{\star}(a_{1:t})-\Xi_{Q}^{(t)}(a_{1:t})
     \bigr\rVert_{1}
  \Bigr]
  &\;\le\;
  \sqrt{\tfrac12\; \mathbb E_{Q^{\star}}\!\left[D\!\bigl(\rho_{E}^{\star}(a_{1:t})\,\big\Vert\,\Xi_{Q}^{(t)}(a_{1:t})\bigr)\right]} \\
  &\;\le\;
  \sqrt{\tfrac12 \cdot \frac{K_{Q}(Q^{\star})\ln 2+\ln(1+g)}{t}} \\
  &\;=\;
  \sqrt{\frac{K_{Q}(Q^{\star})\ln 2+\ln(1+g)}{2t}} \\
  &\;=\;
  \left(\sqrt{\frac{K_{Q}(Q^{\star})\ln 2+\ln(1+g)}{2}}\right) t^{-1/2}.
\end{align*}
Hence :
\begin{align}
  \mathbb E_{Q^{\star}}\!\bigl[
     \tfrac12
     \bigl\lVert
       \rho_{E}^{\star}(a_{1:t})-\Xi_{Q}^{(t)}(a_{1:t})
     \bigr\rVert_{1}
  \bigr]
  =\mathcal O\!\bigl(t^{-1/2}\bigr). \label{eq:trace_distance_convergence_derived}
\end{align}
The significance of Eq.~\eqref{eq:trace_distance_convergence_derived} is that the $\mathcal{O}(t^{-1})$ convergence rate for the expected relative entropy implies an $\mathcal{O}(t^{-1/2})$ convergence rate for the expected trace distance. Thus any measurement strategy can
separate the agent’s belief from the true state only with probability
vanishing like \(t^{-1/2}\). This means that as the agent accumulates more data (as $t$ increases), its belief state $\Xi_Q^{(t)}(a_{1:t})$ not only becomes a better approximation of the true state $\rho_E^\star(a_{1:t})$ in terms of information-theoretic divergence, but also becomes increasingly \emph{indistinguishable} from the true state via any physical measurement. The maximum probability of successfully distinguishing the agent's belief from the true state diminishes as $t^{-1/2}$. This transition from a $1/t$ rate for loss (relative entropy can be seen as a generalised squared error in the space of probability distributions or density matrices) to a $t^{-1/2}$ rate for an $L_1$-like distance (trace distance is the quantum analogue of total variation distance) is a common theme in statistical inference. The $\mathcal{O}(t^{-1/2})$ rate for distinguishability is derived without needing any assumptions beyond those already made for establishing the relative entropy bound (i.e., conditions (C1)–(C3) that underpin Theorem~\ref{thm:qsi-convergence}).
\subsection{Conditions on QSI}
Establishing the convergence of QSI is challenging because of the conditions that would need to subsist:
\begin{enumerate}
    \item \emph{(C1) Ergodicity:} There is a $\delta>0$ such that for every admissible action policy the time-averaged (with respect to each admissible instrument schedule) state of the true environment $Q^\star$ satisfies
    \begin{align*}
        \liminf_{m\to\infty}
        \tfrac1m\sum_{k=1}^{m}
          \bigl\lVert\rho_{E}^{\star}(a_{1:k})-
                         \rho_{E}^{\star}(a_{1:k-1})\bigr\rVert_{1}
      \leq\delta.
    \end{align*}
    This condition implies that the environment's state does not change too erratically on average, ensuring that past observations retain some relevance for predicting future states.
    \item \emph{(C2) Informational Completeness:} Each cycle’s instrument (measurement) has a POVM refinement whose classical Fisher information matrix is $\epsilon$-non-singular (full-rank) up to error $\epsilon>0$. This means the agent's measurements are sufficiently informative (due to the lack of degeneracy) to distinguish different quantum states, which is crucial for learning.
    \item \emph{(C3) Complexity Gap:} The true environment $Q^\star$ has finite quantum Kolmogorov complexity, $K_{Q}(Q^{\star})<\infty$, and the sum of relative complexities of other environments:
    \begin{align*}
        g:=\sum\nolimits_{Q\neq Q^{\star}} 2^{-\bigl(K_{Q}(Q)-K_{Q}(Q^{\star})\bigr)}
    \end{align*}
    is finite. This ensures $Q^\star$ is not infinitely complex relative to the universal QTM and has a non-vanishing initial weight in the QSI mixture, and that the collective weight of alternative hypotheses is manageable.
\end{enumerate}

\subsection{Comparison with classical SI}
In our sketch above, we have adapted the typical classical approach for convergence of classical SI.
\\
\\
\noindent\textit{1. Likelihood Operators and Initial Divergence Bound:}
The QSI mixture $\Xi_{Q}(a_{1:m})$ (Eq.~\eqref{eq:XiQ-density}) is a semi-density operator representing the agent's belief about the environment state before the $m$-th observation, given actions $a_{1:m}$. The posterior $\Xi_{Q}^{(t)}$ is obtained via updates based on observed outcomes (Eq.~\eqref{eq:posteriorupdate}).
A first step, analogous to classical Solomonoff induction, is to bound the initial information deficit of the QSI mixture with respect to the true environment $Q^\star$. This initial divergence, $D_0 = D(\rho_E^\star \| \Xi_Q^{(0)})$, where $\Xi_Q^{(0)}$ is the initial QSI prior, can be bounded by the quantum Kolmogorov complexity of the true environment:
\begin{equation} \label{eq:app-initial-divergence-bound}
    D_0 \le K_{Q}(Q^{\star})\ln 2+\ln(1+g).
\end{equation}
This bound reflects that the cost of encoding the true environment within the universal mixture is related to its complexity $K_Q(Q^\star)$, and the term $\ln(1+g)$ accounts for the mass of alternative hypotheses. Establishing this bound rigorously for semi-density operators and Umegaki relative entropy is central to a full proof.
\\
\\
\noindent\textit{2. Monotonicity of Quantum Relative Entropy and Martingale Argument:}
The core of the convergence argument hinges on the data-processing inequality for quantum relative entropy. This principle states that for any quantum operation (a completely-positive trace-preserving map, CPTP map), such as a measurement branch $\mathcal{M}_{e_k}$ occurring in the Bayesian update, the relative entropy between any two states cannot increase: $D(\mathcal{M}(\rho) \| \mathcal{M}(\sigma)) \le D(\rho \| \sigma)$.
Recall the likelihood operators $\Lambda_{t}^{Q} :=\mathcal M_{e_{t}}\circ\dots\circ\mathcal M_{e_{1}} (\rho_{E}^{Q}(a_{1:t}))$ where $\rho_E^Q(a_{1:t})$ is interpreted as the initial state on which the sequence of CPTP maps $\mathcal{M}_{e_1}, \dots, \mathcal{M}_{e_t}$ act.)
In the main paper we assume that $D(\rho_{E}^{\star}\,\Vert\,\Xi_Q) - D(\Lambda^{\star}_{k}\,\Vert\,\Lambda^{\Xi}_{k}) \;\;\text{is }\mathbb E_{Q^{\star}}\!\text{-martingale}$. Specifically, if $D_k = D(\rho_E^{\star(k)} \| \Xi_Q^{(k)})$ is the divergence between the true state and the posterior after $k$ observation-update cycles, then $\mathbb{E}_{Q^\star}[D_k | \text{history } h_{<k}] \le D_{k-1}$. The divergence, on average, tends to decrease or stay the same.
The decline in divergence at step $k$, $\Delta_k = D_{k-1} - \mathbb{E}_{e_k \sim Q^\star}[D_k | h_{<k}]$, is therefore non-negative. As a result, the assumption $C2$ is needed to ensure that if the posterior $\Xi_Q^{(k-1)}$ is not perfectly aligned with the true state $\rho_E^{\star(k-1)}$. Measurement provides additional information that leads to a decrease in uncertainty (and greater weighting of the QAIXI belief $\Xi$ towards the true state ($g$ declines while $\omega_{Q^*}$ increases).
\\
\\
\noindent\textit{3. Average Divergence:}
By summing these non-negative expected drops over $t$ interaction cycles, we get:
\[ \sum_{k=1}^t \mathbb{E}_{Q^\star, \text{hist}}[\Delta_k] = \mathbb{E}_{Q^\star, \text{hist}}[D_0 - D_t] \le D_0, \]
where $D_0$ is the initial divergence bounded as in Eq.~\eqref{eq:app-initial-divergence-bound}, and $D_t = D(\rho_E^\star(a_{1:t}) \| \Xi_Q^{(t)}(a_{1:t}))$ is the divergence at step $t$.
This implies that the sum of positive one-step learning gains (in terms of divergence reduction) is bounded by the initial total ignorance $D_0$. The result that:
\[ \mathbb{E}_{Q^{\star}}[D_t] \le \frac{D_0}{t} = \frac{K_{Q}(Q^{\star})\ln 2+\ln(1+g)}{t}. \]
shows that the expected divergence at time $t$ itself must decrease on average
This establishes that the expected relative entropy between the true environment state and the QSI posterior converges to zero at a rate of $1/t$.
\\
\\
\noindent\textit{4. Quantum Pinsker Inequality for Trace Distance Convergence:}
Finally, the quantum Pinsker inequality provides a link between the trace distance $\|\rho - \sigma\|_1$ (a measure of distinguishability) and the relative entropy $D(\rho \| \sigma)$:
\[ \frac{1}{2}\|\rho - \sigma\|_1^2 \le D(\rho \| \sigma). \]
Applying this to the result from step 3:
\[ \mathbb{E}_{Q^{\star}}\left[\frac{1}{2} \bigl\lVert \rho_{E}^{\star}(a_{1:t})-\Xi_{Q}^{(t)}(a_{1:t}) \bigr\rVert_{1}^2\right] \le \mathbb{E}_{Q^{\star}}[D_t] \le \frac{K_{Q}(Q^{\star})\ln 2+\ln(1+g)}{t}. \]
Using Jensen's inequality (or a linear version of Pinsker's inequality), this implies convergence in trace distance:
\[ \mathbb{E}_{Q^{\star}}\left[\frac{1}{2} \bigl\lVert \rho_{E}^{\star}(a_{1:t})-\Xi_{Q}^{(t)}(a_{1:t}) \bigr\rVert_{1}\right] = \mathcal{O}(t^{-1/2}). \]
This shows that the QSI posterior state $\Xi_Q^{(t)}(a_{1:t})$ converges to the true environment state $\rho_E^\star(a_{1:t})$ in trace distance. The effect of $\epsilon$-completeness in (C2), rather than perfect informational completeness, would typically introduce an additional error term in the bounds, which vanishes as $\epsilon \to 0$.

\subsection{Limitations}
The proof outline above, while following a known pattern, relies on several steps that require careful and rigorous justification in the quantum setting:

\begin{enumerate}
    \item \textit{Martingale Theory for Quantum Processes}. Classical martingale convergence theorems are well-established for sequences of real-valued random variables with respect to a filtration. Extending these to sequences of quantum states (density operators or, as here, semi-density operators like $\Xi_Q$) that evolve via quantum operations (measurements and unitary evolutions) is a significant technical hurdle. The notion of a filtration $\mathcal{F}_t$ must properly capture the information gained from sequences of quantum measurements, which is complicated by measurement back-action, the non-commutative nature of observables, and potential contextuality.
    \item \textit{Relative Entropy with Semi-Density Operators}. The Umegaki relative entropy $D(\rho\|\sigma)$ is typically defined for density operators $\rho$ and $\sigma$ (where $\Tr\rho = \Tr\sigma = 1$). We have set out how this would be adjusted (normalised) in order to render it a full density operator. An alternative approach given the QSI mixture $\Xi_Q$ is a semi-density operator ($\Tr \Xi_Q \le 1$) is for the derivation - and properties (like monotonicity under CPTP maps), and the Pinsker inequality - to be established or appropriately adapted for semi-density operators.
    \item \textit{Measurement Back-Action}. Each measurement $\mathcal{M}_{e_k}$ provides a classical outcome $e_k$ and transforms the quantum state. The sequences of true states $\rho_E^\star(a_{1:k})$ and posterior beliefs $\Xi_Q^{(k)}(a_{1:k})$ are conditioned on the full history of actions $a_{1:k}$ and prior outcomes $e_{1:k-1}$. The calculation of the expected drop in divergence requires careful averaging over outcomes $e_k$ generated by the true environment $Q^\star$ acting on its state $\rho_E^\star(a_{1:k-1})$. The ergodicity condition (C1) is important in ensuring that the learning process is meaningful over time and avoids highly non-stationary state sequences.
    \item \textit{Contextuality and Non-locality}. As we discuss, if the true environment $Q^\star$ exhibits Kochen-Specker contextuality or Bell non-locality, the classical intuition that a history uniquely determines future outcome probabilities (independent of measurement choices for other compatible observables) breaks down. The requirement for an adapted, context-aware filtration complicates the structure of the conditional expectations inherent in any martingale argument e.g. standard assumptions of conditional independence of percepts, often used in classical proofs, are violated. However, the practical implications of these foundational issues may be more or less significant depending on context (see the Appendix below for an extended discussion).
    \item \textit{Informational Completeness}. Condition (C2) asserts that measurements are sufficiently informative (Fisher information matrix is full-rank up to $\epsilon$), the precise manner in which this guarantees a sufficient trend towards convergence at each step would require further analysis. For example, how (and whether) the error $\epsilon$ propagates into the final convergence bounds and rates is would be open question.
    \item \textit{No-Cloning and Sample Complexity}. The no cloning theorem fundamentally limits the agent's ability to learn. A QAIXI agent cannot make multiple measurements on identical copies of a past unperturbed quantum state to refine its likelihood estimates for different actions from that state as measurement consumes the state. This underscores the online and single-shot nature of quantum learning from a trajectory. It also has consequences for the effective sample complexity of learning undertaken by any QAIXI and is a significant difference from classical scenarios where data can often be re-analysed. Note that if i.i.d. copies at negligible cost, the no-cloning impact may be minimal.
\end{enumerate}
While the QSI convergence theorem sketch outlines a potentially plausible route to proving that a QAIXI-like agent can learn its quantum environment, a full proof faces unresolved theoretical and technical challenges. These primarily arise from the foundational differences between classical and quantum information processing, particularly concerning measurement, state representation, and contextuality. Addressing these issues remains a significant open research area in the theoretical foundations of quantum artificial general intelligence. QAIXI proposals can assist research into intelligent systems by clarifying how quantum phenomena may affect idealised agents. However, as noted above, in addition to the significant practical barriers to QAIXI, there remain a considerable number of open theoretical questions regarding the extent to which QAIXI proposals, including those above, are well-founded or feasible. These include whether the relationship of concepts of intelligence and quantum information processing is anything other than instrumental as a tool for classical agents.


\newpage
\section{Foundational Implications for QAIXI}
\label{app:ks-contextuality}

The Kochen-Specker (KS) theorem \cite{kochenSpecker1967} is a cornerstone of quantum foundations, reflecting central differences between quantum and classical descriptions of reality. It states that for quantum systems of dimension three or higher, it is impossible to assign definite, pre-existing values to all possible measurements (observables) in a way that is independent of the context of measurement (i.e., the order of measurement operations and which other compatible observables are measured alongside it). Here we elaborate in more detail on the implications of KS contextuality for the QAIXI framework, particularly concerning Theorem~\ref{thm:contextual_predictor_full}.

\subsection{Contextuality and QAIXI's Universal Mixture} \label{app:contextuality}
Classical AIXI operates under the assumption that an environment $\nu$ has definite properties that can be observed. The history $h_{<t}$ is a sequence of such definite observations. However, if the true quantum environment $Q^{\star}$ leads to states whose measurements exhibit contextuality, then QAIXI's understanding of $Q^*$ will differ from the classical case. Recall Theorem~\ref{thm:contextual_predictor_full} provides that no quantum Turing machine $Q$ (representing a predictive model for the environment) can output a commuting family of projectors $\{Q_{a_{1:m}}(e_{1:m})\}$ that perfectly predicts outcomes for a KS-uncolourable set of measurements (Eq.~\eqref{eq:perfect_prediction}). The consequence for QAIXI's belief state, $\Xi_{Q}(a_{1:m})$ (defined in Eq.~\eqref{eq:XiQ-density}), is significant. QAIXI's belief is a mixture that sums over all semi-computable quantum environments $Q \in \mathcal{Q}_{\!\mathrm{sol}}$. If $Q^\star$ is KS-contextual, then any individual $Q$ within the sum that attempts to model $Q^\star$ using classical, non-contextual hidden variables (or by outputting commuting projectors that aim to assign definite values independent of measurement choices) will fail to accurately reproduce the statistics of $Q^\star$. Thus, for $\Xi_Q$ to converge to a description of $Q^\star$, it must implicitly give dominant weight to models $Q$ that are themselves contextual or whose predictive mechanisms do not rely on non-contextual value assignments.

\subsection{Contextuality Consequences for Learning}
In classical AIXI, the history $h_{<t}$ is a sequence of action-percept pairs, where percepts are assumed to be objective records of definite environmental properties via measurement. For QAIXI situated within a contextual quantum environment, the situation is more nuanced. The outcome $o_t$ of a measurement action $a_t$ may not be interpretable as revealing a pre-existing property $P_j$ of the system. Instead, $o_t$ is realised only upon measurement. Intuitively, its probability (and even its meaning) can depend on the complete set of compatible observables measured in action $a_t$.
To see this, let $Q^\star$ be an environment that prepares a qutrit ($d=3$) system. Suppose QAIXI can choose actions $a_t$ that correspond to measuring sets of compatible projectors from a KS set (e.g., a set of projectors for a qutrit as in the Peres-Mermin square \cite{peres1991two}, or simpler KS sets).
Let $P_1, P_2, \dots, P_n$ be such a KS set.
\begin{itemize}
    \item If QAIXI performs action $a_1$ measuring the context $\mathcal{C}_1 = \{P_i, P_j, P_k\}$ (where $P_i+P_j+P_k = \mathbb{I}$), it obtains outcomes $(o_i, o_j, o_k)$.
    \item If it later (or in a counterfactual scenario) performs action $a_2$ measuring context $\mathcal{C}_2 = \{P_i, P_l, P_m\}$ (where $P_i+P_l+P_m = \mathbb{I}$), it obtains outcomes $(o'_i, o'_l, o'_m)$.
\end{itemize}
The KS theorem implies that QAIXI cannot learn a universal value assignment $v(P_x) \in \{0,1\}$ such that $o_i = v(P_i)$ and $o'_i = v(P_i)$ consistently across all contexts for all projectors in the KS set, while also satisfying $\sum_{P_x \in \mathcal{C}} v(P_x) = 1$ for all contexts $\mathcal{C}$. This means that the Bayesian update rule for QAIXI (Eq.~\eqref{eq:posteriorupdate}), $\Xi_{Q}^{(t)}\!(a_{1:t-1}) := \frac{\mathcal M_{e_{t-1}}\!\bigl(\Xi_{Q}^{(t-1)}(a_{1:t-2})\bigr)}{\dots}$, must process the information $e_{t-1}$ (which includes $o_{t-1}$) in a way that acknowledges its contextual nature. The meaning of $o_{t-1}$ for updating beliefs about $Q^\star$ is tied to the full instrument $\mathcal{M}_{e_{t-1}}$ and potentially the set of all measurements that constituted action $a_{t-1}$. This is why the posterior must be refined by the entire future instrument schedule (or at least the current one) if one aims for precise predictions in contextual scenarios. A simple classical history string $h_{<t}$ is insufficient.

\subsection{Convergence Implications}
The conditions for the convergence of QSI as sketched in Theorem~\ref{thm:qsi-convergence}, relies on a martingale argument. Martingales are defined with respect to a filtration, which represents the information accumulated over time. If measurement outcomes are contextual, the structure of this filtration becomes more complex. The information gained from an outcome $o_t$ is not just about the specific observable measured but also about the context.
This is why the martingale assumptions of Theorem~\ref{thm:qsi-convergence} require an adapted filtration that records measurement contexts - because standard conditional independence assumptions, often implicit in simpler martingale proofs, may not hold if the probability of future outcomes depends on the context of past measurements in non-trivial ways. The ergodicity (C1) and informational completeness (C2) conditions would also need to be interpreted in light of contextuality. In essence, KS contextuality reflects the differences between classical and quantum ontology - and classical and quantum models of computation. Its consequence is that any QAIXI would need to adopt a more sophisticated model of reality than classical AIXI. It cannot assume that the environment possesses a set of definite, non-contextual properties that are merely uncovered by measurement. Instead, QAIXI must learn and operate in a world where measurement outcomes are co-created by the interaction between the agent's choice of measurement context and the quantum system.

\section{Bell Non-Locality and No-Cloning}
\label{app:bell-no-cloning}
Beyond contextuality, other foundational quantum principles like Bell non-locality and the no-cloning theorem impose significant constraints and offer unique characteristics to QAIXI compared to its classical counterpart. We expand upon the discussion above in the subsections below.

\subsection{Bell Non-Locality}
Bell's theorem \cite{bell_speakable_2004} demonstrates that quantum mechanics predicts correlations between spatially separated systems (that were previously entangled) which cannot be explained by any theory based on local hidden variables (LHVs). The implications for QAIXI are indicative of the differences between quantum and classical environments.  If the true environment $Q^{\star}$ subsists in entangled states and allowing QAIXI to perform measurements on these subsystems at space-like separation, then QAIXI's universal mixture $\Xi_Q$ (Eq.~\eqref{eq:XiQ-density}) must accommodate models $Q$ that are inherently non-local:
\begin{enumerate}
    \item Classical AIXI, relying on Solomonoff induction over classical Turing machines $\Mcal_{sol}$, constructs its universal prior $\xi_U$ (Eq.~\eqref{eq:xi_U_classical_redef}) from environments that are, by their classical nature, local. Such a prior would assign zero probability to observing correlations that violate Bell inequalities (e.g., the CHSH inequality).
    \item QAIXI, by summing over quantum environments $\mathcal{Q}_{\!\mathrm{sol}}$, can, in principle, learn and adapt to a non-local $Q^\star$. The term $2^{-K_Q(Q)}\rho_E^Q(a_{1:m})$ must include $Q$s whose $\rho_E^Q(a_{1:m})$ can lead to Bell-violating statistics upon appropriate measurements.
\end{enumerate}
For example, consider a QAIXI agent designed with two components, Alice and Bob, who are spatially separated.
\begin{enumerate}[label=(\alph*)]
    \item The environment $Q^\star$ repeatedly sends Alice and Bob a pair of qubits in an entangled Bell state, e.g., $|\Psi^-\rangle = \frac{1}{\sqrt{2}}(|01\rangle - |10\rangle)$.
    \item At each cycle $t$, Alice receives a classical random bit $x \in \{0,1\}$ and Bob receives $y \in \{0,1\}$ (these could be part of the agent's internal state or from an external prompter, effectively part of the action setup).
    \item Alice chooses a measurement setting $A_x$ (e.g., measuring her qubit along one of two directions) and obtains outcome $o_A \in \{+1, -1\}$. Bob similarly chooses $B_y$ and gets $o_B \in \{+1, -1\}$.
    \item The action for QAIXI could be $a_t = (\text{setting choice for } A_x, \text{setting choice for } B_y)$, and the percept $e_t = (o_A, o_B, r_t)$. The reward $r_t$ could be $1$ if $o_A \cdot o_B = (-1)^{xy}$ (CHSH game condition for certain settings) and $0$ otherwise.
\end{enumerate}
If $Q^\star$ is quantum, Alice and Bob can choose their measurement settings such that they win the CHSH game with a probability that is impossible for any classical LHV strategy.
A classical AIXI, whose models $\nu \in \Mcal_{sol}$ are constrained by LHV, would never be able to predict or achieve this quantum level of success. QAIXI, with its quantum prior $\Xi_Q$, could learn to implement the optimal quantum strategy and understand that $Q^\star$ is non-local. The presence of non-local correlations means that the percept $e_t=(o_A, o_B, r_t)$ contains components $o_A$ and $o_B$ that are correlated in a way not explainable by any shared information in their common past light-cone (beyond the initial entanglement). This challenges classical notions of conditional independence typically used in agent-based convergence proofs. It can also break the martingale structure assumed in Theorem \ref{thm:qsi-convergence} unless non-local models $Q$ are included in the mixture. More detail can be found in literature on quantum game theory \cite{gutoski2007toward,guo2008survey,bostanci2022quantum}. Moreover, recent literature has shown that in principle quantum machines could learn from exponentially less experiments than classical analogues by leveraging non-local correlations \cite{huang2022quantum}. Thus the impact of non-locality on learning of any quantum AIXI agent remains an important and open question.

\subsection{No-Cloning Theorem and QAIXI's Sample Complexity} \label{app:no-cloning-sample-complexity}
The no-cloning theorem states that it is impossible to create an identical copy of an arbitrary unknown quantum state. This has consequences for how QAIXI can learn about its environment. 

\subsubsection{Learning} Firstly, because percepts upon which QAIXI learning is based are QTC channels, the no cloning theorem gives rise to differences in how QAIXI learning would occur:
\begin{enumerate}
    \item When QAIXI performs an action $a_t$ (especially if it's a measurement) on the environment state $\rho_E^{Q}(a_{1:t})$, this interaction alters the state. The post-measurement state is different, and the specific instance of $\rho_E^{Q}(a_{1:t})$ is consumed in yielding the outcome $o_t$.
    \item Classical AIXI can, in principle, take a history $h_{<t}$ and test many hypothetical continuations with a given model $\nu$ without altering the data $h_{<t}$. Classical information can be copied and reused.
    \item QAIXI cannot do this with quantum states. To know how $Q^\star$ would have responded to a different action $a'_t$ from the \emph{exact same} instance of $\rho_E^{\star}(a_{1:t-1})$, is impossible. It would need $Q^\star$ to produce that state (or an identical one) again.
\end{enumerate}

\subsubsection{Sample Complexity}
This one-shot nature of quantum measurement on unknown states impacts the sample complexity of learning by a QAIXI agent:
\begin{enumerate}
    \item To distinguish between different hypotheses $Q$ about the environment, or to estimate the expected outcome/reward for different actions, QAIXI needs to observe the environment's response multiple times.
    \item Since each observation of a specific state instance is unique and unrepeatable, QAIXI effectively needs $Q^\star$ to prepare a new instance of a comparable state for each piece of information it wants to gather about a particular type of situation or action.
    \item If QAIXI wants to learn the full characteristics of $\rho_E^{\star}(a_{1:t})$ through measurements, it is performing a form of quantum state tomography. Full tomography of an $n$-qubit state generally requires a number of measurement settings and repetitions that scales exponentially with $n$. While QAIXI's goal is not necessarily full tomography but rather optimal action selection, its ability to learn the relevant features of $Q^\star$ is still constrained by the information extractable per (unclonable) interaction.
    \item As noted above, this means the learning rate is limited by state-preparation resources. If $Q^\star$ itself has computational costs or time delays associated with preparing states, this directly translates into a slower learning rate for QAIXI in terms of real-time or computational steps of $Q^\star$.
\end{enumerate}
The informational completeness condition (C2) for Theorem~\ref{thm:qsi-convergence} ensures that measurements are informative. However, no-cloning dictates that this information is gathered sequentially, one unclonable sample at a time. The $1/t$ convergence rate for relative entropy and $\mathcal{O}(t^{-1/2})$ for trace distance are asymptotic statements about the number of interactions $t$. The practical time or resources needed to achieve a certain level of accuracy will be higher in quantum scenarios where each sample is precious and unrepeatable, compared to classical scenarios where data can be exhaustively analyzed. Bell non-locality thus expands the class of environments QAIXI must consider beyond classical capabilities, while the no-cloning theorem imposes a fundamental restriction on the efficiency with which QAIXI can extract information from its quantum world, directly impacting its sample complexity for learning.

\newpage
\section{Notation}
\begin{table}[ht!]
\centering
\renewcommand{\arraystretch}{1.12}
\begin{tabularx}{\linewidth}{>{\bfseries}l X}
\toprule
\multicolumn{2}{c}{\textbf{Glossary of Symbols}}\\
\midrule
$a_t$                         & Action chosen by the agent at cycle $t$ (classical label). \\
$o_t,\;r_t$                   & Observation and reward components of the percept. \\
$e_t=(o_t,r_t)$               & Percept at cycle $t$.
\\
$w_\nu$               & Weighted posterior: $w_\nu = \frac{2^{-K(\nu)}\,
        \prod_{i=1}^{t-1}\nu(e_i\mid h_{<i}a_i)}
       {\displaystyle
        \sum_{\mu\in\Mcal_U}2^{-K(\mu)}\,
        \prod_{i=1}^{t-1}\mu(e_i\mid h_{<i}a_i)}$.\\
$h_{<t}$                      & Complete history up to but not including cycle $t$. \\
$\mathcal A,\;\mathcal O$     & Action and observation alphabets. \\
$m$                           & Fixed lifetime / planning horizon. \\
$\gamma$                      & Discount factor in $[0,1)$. \\
\midrule
$\mu,\nu$                     & Classical (semi-computable) environments. \\
$K(\nu)$                     & Classical Kolmogorov complexity of $\nu$. \\
$\xi_U$                       & Classical Solomonoff prior $\displaystyle\sum_{\nu}2^{-K(\nu)}\nu(\,\cdot\,)$. \\
\midrule
$\mathcal H_A,\;\mathcal H_E$ & Hilbert spaces of the agent and environment registers. \\
$\rho_A^{(t)},\rho_E^{(t)}$   & Agent / environment density operators for cycle $t$. \\
$\rho_{AE}^{(t)}$             & Joint state of agent and environment. \\
$\Phi_{U_{a_t}}$              & Unitary CPTP channel implementing a coherent action $a_t$. \\
$\mathcal I_{a_t}=\{\mathcal E^{a_t}_k\}_k$ & Quantum instrument representing a measurement action. \\
$\Gamma_{\mathrm{obs}}$       & Finite outcome alphabet of the instrument. \\
$M^{a_t}_k$                   & Kraus / POVM operators of branch $k$. \\
\midrule
$\mathcal Q_{\!\mathrm{sol}}$ & Set of chronological, semi-computable quantum environments. \\
$Q$                           & Individual quantum environment (CPTP channel family). \\
$|Q\rangle$                   & Choi–Jamiołkowski purification of an environment channel. \\
$K_Q(Q)$                      & Quantum Kolmogorov complexity (Eq.~\ref{eq:QKC}). \\
$\Xi_Q(a_{1:m})$              & \emph{Operator-valued} universal mixture (semi-density op.) \\
$\xi_Q(e_{1:m}\!\Vert a_{1:m})$ & \emph{Scalar} probability obtained by projecting $\Xi_Q$ onto the instrument POVM. \\
$U_{\text{univ}}$             & Fixed universal quantum Turing machine. \\
\midrule
$V^{\pi}_{Q}$                 & Discounted value of policy $\pi$ in environment $Q$. \\
$V^{\pi}_{\Xi_Q}$             & Value averaged under the universal mixture (Eq.~\ref{eq:qaixi‐value}). \\
$\pi^{\mathrm{QAIXI}}$        & Optimal policy that maximises $V^{\pi}_{\Xi_Q}$. \\
\midrule
$\Lambda_t^{Q}$               & \emph{Likelihood operator}: state after applying the first $t$ measurement branches to $\rho_E^{Q}$ (Sec.~\ref{sec:qsi}). \\
$D(\rho\Vert\sigma)$          & Umegaki quantum relative entropy $\operatorname{Tr}[\rho(\log\rho-\log\sigma)]$. \\
\bottomrule
\end{tabularx}
\caption{Summary of QAIXI notation.}
\label{tab:symbols}
\end{table}

\end{document}